\documentclass[12pt,draftcls,onecolumn]{IEEEtran}
\usepackage[utf8]{inputenc}
\usepackage{amsmath,mathtools}
\usepackage{amsthm}
\usepackage{amsfonts}
\usepackage{amssymb}
\DeclareUnicodeCharacter{00A0}{ }
\usepackage{graphicx}
\usepackage{graphics}
\usepackage{float}
\usepackage{tikz}
\usepackage{booktabs, makecell}
\usepackage{caption}
\usepackage{graphicx}
\usepackage{blindtext}
\usepackage{booktabs}
\usepackage[rightbars]{changebar}
\usetikzlibrary{shapes,snakes,patterns}
\usepackage{epstopdf}
\usepackage[normalem]{ulem}
\renewcommand{\vec}[1]{\mathbf{#1}}
\usepackage[justification=centering]{caption}
\usepackage{multirow}
\usepackage{multicol}
\usepackage{enumerate} 
\usepackage{cite}
\usepackage{placeins}
\newcommand{\vast}{\bBigg@{4}}
\newcommand{\Vast}{\bBigg@{5}}
\newtheorem{theorem}{Theorem}[]
\newtheorem{lemma}[]{Lemma}
\usepackage{suffix}
\usepackage{mathtools}
\usepackage{amsthm}
\usepackage{enumitem}
\usepackage{xfrac}
\usepackage{relsize}
\usepackage{dsfont}
\usepackage{url}
\usepackage{xcolor}
\usepackage{slashbox}
\usepackage{breqn}
\theoremstyle{definition}

\newtheorem{corollary}{Corollary}[theorem]
\usepackage[list=true]{subcaption}
\usepackage{authblk}
\allowdisplaybreaks
\begin{document}
\bstctlcite{IEEEexample:BSTcontrol}
\title{Outage Probability Expressions for an IRS-Assisted System with and without Source-Destination Link for the Case of Quantized Phase Shifts in $\kappa - \mu$ Fading}
\author{Mavilla Charishma$^{*}$, Athira Subhash$^{*}$, Shashank Shekhar$^{*}$ and Sheetal Kalyani  \footnote{Mavilla Charishma, Athira Subhash, Shashank Shekhar and Sheetal Kalyani are with the Department of Electrical Engineering, Indian Institute of Technology Madras. (email:\{ee19m020@smail,ee16d027@smail,ee17d022@smail,skalyani@ee\}.iitm.ac.in).\\$^*$Mavilla Charishma, Athira Subhash and Shashank Shekhar are co-first authors.}
}
\maketitle 
\section*{Abstract}
\cbstart \textcolor{blue}{In this work, we study the outage probability (OP) at the destination of an intelligent reflecting surface (IRS) assisted communication system in a $\kappa-\mu$ fading environment. A practical system model that takes into account the presence of phase error due to quantization at the IRS when a)  source-destination (SD)  link is present and b) SD link is absent is considered.} \cbend First, an exact expression is derived, and then we derive three simple approximations for the OP using the following approaches: (i) uni-variate dimension reduction, (ii)  moment matching and, (iii) Kullback–Leibler divergence minimization. The resulting expressions for OP are simple to evaluate and quite tight even in the tail region. The validity of these approximations is demonstrated using extensive Monte Carlo simulations. \cbstart \textcolor{blue}{We also study the impact of the number of bits available for quantization, the position of IRS with respect to the source and destination
and the number of IRS elements on the OP for systems with and without an SD link.} \cbend
\begin{IEEEkeywords}
Intelligent reflecting surface, phase error, outage probability, multivariate integration, Kullback–Leibler divergence, uni-variate dimension reduction.
\end{IEEEkeywords}
\section{Introduction}
Intelligent reflecting surface (IRS) assisted communication has gained much research momentum recently \cite{lis_renzo_2020}. Reconfigurable IRS realized using arrays of passive antenna elements or scattering elements made from metamaterials can introduce specific phase shifts on the incident electromagnetic signal without any decoding, encoding, or radiofrequency processing operations \cite{basar2019,Zhang5g}. \cbstart \textcolor{blue}{Appropriate phase shifts introduced on a large number of reflector elements can create a coherent combination of their individually scattered signals, which is narrow and focused at the user \cite{Zhang5g}. This strategy is known as energy focusing \cite{Zhang5g,wu2019towards}.} \cbend Traditionally, reliable communication links were achieved by implementing intelligent transmitter and receiver designs that combat signal deterioration's introduced by the propagation environment. However, in the past few years, there has been a shift in this paradigm towards the idea of smart radio environments (SRE) where the performance gains achievable via `smartly' modifying the wireless propagation channel are explored \cite{lis_renzo_2020}. 
\par Several works in literature study the performance of SRE realized using IRS and compare them with the performance achievable using techniques like cooperative relaying, massive multiple-input multiple-output (M-MIMO), distributed antennas, backscatter communication, millimetre (mm)-wave communication, and network densification \cite{renzo_2020,wang_spawc_2020,9348003,gopi2020intelligent,wang2020intelligent}. The authors of \cite{wang_spawc_2020} show that an  IRS-assisted M-MIMO system can use the same channel estimation overhead as an M-MIMO system with no IRS to achieve a higher user signal to interference and noise ratio (SINR). The authors of \cite{sena_oct_2020} study the performance gains achieved by an IRS-assisted M-MIMO system integrated with a non-orthogonal multiple access (NOMA) network. They also discuss the critical challenges in realizing an IRS-aided NOMA network. The key differences and similarities between an IRS-aided network and a relay network are studied by the authors of \cite{renzo_2020}. Using mathematical analysis and numerical simulations, they demonstrate that a sufficiently large IRS can outperform relay-aided systems in terms of data rate while reducing the implementation complexity. The number of reflector elements required to outperform the performance of the decode and forward (DF) relay system is studied by the authors of \cite{bjornson2019intelligent}. IRS-aided bistatic backscatter communication (BackCom)  system is studied in \cite{9348003}. Notably, the joint optimization of the phase shifts at the  IRS  and the transmit beamforming vector of the carrier emitter that minimizes the transmit power consumption is studied. Another exciting venue where IRS has useful applications are the mm-wave communication systems \textcolor{blue}{\cite{gopi2020intelligent,Zhangmmwave,wang2020intelligent}}. Authors of \cite{gopi2020intelligent} proposed three different low-cost architectures based on IRS for beam index modulation schemes in mm-wave communication systems. These schemes are capable of eliminating the line-of-sight blockage of millimetre wave frequencies. \cbstart \textcolor{blue}{ The authors of \cite{Zhangmmwave} modeled the small scale fading in an mm-wave communication system using the fluctuating two-ray (FTR) distribution and studied the outage probability (OP) and average bit error rate performance of both an IRS aided system and an amplify and forward (AF) relay-based system. They demonstrate the superiority of the IRS aided systems over AF systems in the mm-wave regime even with a small number of reflecting elements}. \cbend The authors of \cite{wang2020intelligent} also study the prospects of combating issues in mm-wave systems like severe path loss and blockage using an IRS to provide effective reflected paths and hence enhance the coverage. 
\par Most of the works discussed above consider performance metrics such as OP, rate and spectral efficiency to evaluate the performance of the IRS-aided communication network. \textcolor{blue}{Several works} in the literature propose different approximations for the OP of an IRS aided system involving one source and one destination node. The authors of \cite{tao2020performance} assume the availability of a large number of reflector elements and hence use the central limiting theorem (CLT) to derive an approximate expression for the OP. Similarly, the authors of \cite{kudathanthirige_icc_20} also derive an approximation for the OP using CLT. Unlike the system model in \cite{tao2020performance}, they assume that the direct link between the source and the destination is in a permanent outage. Such Gaussian approximations are also used to characterize metrics like ergodic capacity, secrecy outage probability in many other works, including \cite{xu2020ergodic,yang2020secrecy,tang2020physical,li2020ergodic}. Similarly, Gamma approximations (using moment matching) are used for deriving approximate OP by the authors of \cite{atapattu2020reconfigurable,de2021large,hou2019mimo,9369134}. As mentioned by the authors of \cite{9369134}, the Gamma distribution is a  Type-III  Pearson distribution and is widely used in fitting distributions for positive random variables (RVs) \cite{al2010approximation,srinivasan2018secrecy}. The authors of  \cite{xu2020ergodic,li2020ergodic,wang2020chernoff,de2021large} consider more practical IRS models, where due to hardware constraints, the possible phase shifts at the IRS elements are restricted to a finite set of discrete values. \cbstart \textcolor{blue}{Most of the above works consider Rayleigh fading channels for analytical tractability. However, the IRS can be deployed at a height such that a LOS link may be available with the source or the destination or both. The authors of \cite{tao2020performance} consider one such scenario where both the links with the source and the destination experience Rician fading. In this work, we consider the more general scenario where all the links experience independent $\kappa-\mu$ fading. Most of the common fading models like the Rayleigh, Rician, Nakagami$-m$ can be derived as special cases of the $\kappa-\mu$ fading model \cite{yacoub2007kappa}.} \cbend 
\par Table \ref{literature_summary} provides a brief summary of the critical literature on this topic\footnote{Note that the authors of \cite{9369134} also use Gamma approximation. However, they consider a different path loss model for a system without an SD link and hence, neither we include \cite{9369134} in Table \ref{literature_summary}, nor do we compare our performance with their results.}. Here, the antenna model refers to the antenna model of the source and destination pair devices. From the table, it is apparent that the study of OP considering both $b$ bit phase quantization at the IRS and an active source-destination (SD) link is not available in the open literature. 
\begin{table}
\centering
\begin{tabular}{|l|l|l|l|l|l|l|}
\hline
 Reference & \begin{tabular}[c]{@{}l@{}}Exact/\\ approximate\end{tabular} & \begin{tabular}[c]{@{}l@{}}Kind of \\ approximation\end{tabular} & Impairment & \begin{tabular}[c]{@{}l@{}}Antenna\\  model\end{tabular}                  & \begin{tabular}[c]{@{}l@{}}SD link\end{tabular}  & \textcolor{blue}{Fading}         \\  \hline
 \cite{tao2020performance},\cite{kudathanthirige2020performance} & Approximate & \begin{tabular}[c]{@{}l@{}}CLT and hence \\ gaussian \\ approximation\end{tabular} &No & SISO & \begin{tabular}[c]{@{}l@{}}\cite{tao2020performance}:Yes \\  \cite{kudathanthirige2020performance}: No\end{tabular} & \begin{tabular}[c]{@{}l@{}}\textcolor{blue}{Rician}, \\  \textcolor{blue}{Rayleigh}\end{tabular}  \\ 
\hline
 \cite{atapattu2020reconfigurable},\cite{hou2019mimo} & Approximate                                                  & \begin{tabular}[c]{@{}l@{}}Gamma moment\\ mathching\\ for sum of \\ double Rayleigh\end{tabular}                       & No         & SISO                                                                      & \begin{tabular}[c]{@{}l@{}}\cite{hou2019mimo}:Yes \\ \cite{atapattu2020reconfigurable}: No \end{tabular}               & \textcolor{blue}{Rayleigh}                                          \\ \hline
  \cite{wang2020study},\cite{de2021large}     & Approximate                                                  & \begin{tabular}[c]{@{}l@{}}Gamma approximation\\  \cite{wang2020study}:for double-rayleigh,\\\cite{de2021large}: for SNR\end{tabular} & \begin{tabular}[c]{@{}l@{}}Yes, quantisa-\\ tion error for \\ \cite{wang2020study}: 1 bit phase \\ \cite{de2021large}: $b$ bit phase \end{tabular} & SISO                    & \begin{tabular}[c]{@{}l@{}}\cite{wang2020study}:Yes \\ \cite{de2021large}: No \end{tabular}                        & \textcolor{blue}{Rayleigh}                                      \\
\hline
 This work  & Approximate & \begin{tabular}[c]{@{}l@{}}(a) Approximation for \\ exact integral \\ for outage \\ (b) Gamma Moment \\ matching  \\ (c) Gamma KL\\ divergence min\end{tabular}  &  \begin{tabular}[c]{@{}l@{}}Yes,\\ quantisation \\ error for\\ $b$ bit phase\\ representation \end{tabular} & SISO & Yes & \color{blue}$\kappa-\mu$ \color{black} \\ \hline 
\end{tabular}
\caption{Key literature studying the OP of IRS-assisted communication systems.}
\label{literature_summary}
\end{table}
\par \cbstart \textcolor{blue}{In this work, we present an exact expression and three different approximations for characterizing the OP at the destination of an IRS-assisted communication system in a $\kappa-\mu$ fading environment.} \cbend Here, we consider a practical scenario where the phase shift at the IRS elements only takes a finite number of possible values owing to the quantization of the phase at the IRS. Furthermore, we evaluate the system's performance both in the presence and absence of a direct link between the source and the destination node. Our major contributions are summarised as follows:
\cbstart
\begin{itemize}
\item \textcolor{blue}{We study the OP of an IRS assisted system in the presence of phase error due to quantization at the IRS in a $\kappa-\mu$ fading environment.}
\item We derive an exact expression for the OP in terms of a multi-fold integral and approximate it using the uni-variate dimension reduction method. 
\item Using the method of moment matching\footnote{The authors of \cite{de2021large} and \cite{atapattu2020reconfigurable} also uses moment matching to obtain a Gamma approximation, however in scenarios without an SD link and without quantization error, respectively. Our approximation for the received SNR is for a more general scenario, and it recovers the result in \cite{de2021large} as a special case.}, we approximate the received SNR as a Gamma RV and hence derive a simple expression for the corresponding OP. 
\item  We also derive the parameters of the Gamma distribution that has the least Kullback–Leibler (KL) divergence with the exact distribution of SNR\footnote{The Gamma distribution, which has minimum KL divergence with respect to the distribution of the received SNR cannot be obtained by moment matching.}. We thus characterize the OP in terms of the cumulative distribution function (CDF) of the resulting Gamma RV.
\item \textcolor{blue}{We also observed the impact of the number of IRS elements, location of the IRS and the number of bits available for quantization on the OP for an IRS aided system with and without an SD link.}
\end{itemize}
\cbend 
\subsubsection*{Organization} The rest of the paper is organized as follows. The system model we consider is presented in Section \ref{model}. Next, in Section \ref{sec_op}, we propose three approximations to evaluate the OP. In Section \ref{simulation}, we verify the utility of our expressions through simulation experiments and  present insights regarding the impact of various system parameters on the OP and finally, Section \ref{conclusion} concludes the work.
\subsubsection*{Notation}
\textcolor{blue}{Here, $\mathbf{S}$, and $\mathbf{D}$ denotes the source, and the destination, respectively.}
symmetric complex Gaussian random variable with mean zero and variance $\sigma^2$.  $\text{diag}(a_{1},\cdots,a_{N})$ denotes a diagonal matrix with entries $a_{1},\cdots,a_{N}$ and $\text{arg(z)}$ denotes the argument (phase) of the complex number $z$.
\section{System model}
\label{model}
 \begin{figure}[ht]
    \centering
    \begin{tikzpicture}[scale=0.8]
\node at(2,1.4) {$\mathbf{h}^{SR}$};
\node at(6.6,1.4) {$\mathbf{h}^{RD}$};
\node at(4.2,-0.3) {$h^{SD}$};
    \draw[->] (0.2, 0) -- (8.8, 0) ;
\draw[->] (0.2,0.2) -- (3.8,1.8){} ;
\draw[->] (4.2,1.8) -- (8.8,0.2) ;
\draw(2,2) rectangle(6,4);
\draw[step=0.2cm,gray,very thin](2,2) grid(6,4);
    \node at (0, 0) [circle,fill,scale=0.5]{};
    \node at (9, 0) [circle,fill,scale=0.5]{};
    \node at (0, -0.5) []{$\mathbf{S}$};
    \node at (9, -0.5) []{$\mathbf{D}$};
    \node at (4, 3) []{IRS};
    \put(0.2,0.5){books}
    \end{tikzpicture}
     \caption{System Model}
      \label{system_model}
\end{figure}
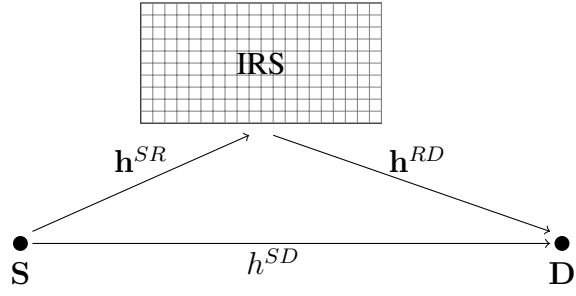
\par We consider a system consisting of one source node ($\mathbf{S}$) communicating with one destination node ($\mathbf{D}$) using an  $\mathbf{IRS}$ with $N$ reflector elements as shown in Fig. \ref{system_model}. Here $\mathbf{S}$ and $\mathbf{D}$  are both equipped with a single antenna each. Furthermore, we assume that the distance between $\mathbf{IRS}$  and $\mathbf{S/D}$ is large enough such that all elements of the $\mathbf{IRS}$  are at the same distance from $\mathbf{S/D}$.
Let,
$\vec{h}^{SR} \in  \mathbb{C}^{N\times1}$, $\vec{h}^{RD} \in \mathbb{C}^{N\times1}$
and
${h}^{SD} \in  \mathbb{C}^{1}$ denote the small-scale fading channel coefficients of the $\mathbf{S}$ to $\mathbf{IRS}$, $\mathbf{IRS}$ to $\mathbf{D}$ and $\mathbf{S}$ to $\mathbf{D}$ link respectively. \cbstart  \textcolor{blue}{It is assumed that all the channels experience independent $\kappa-\mu$ fading, i.e $[\vec{h}^{AB}]_n$ follows the $\kappa-\mu$ distribution with parameters $\kappa_{AB}$ and $\mu_{AB}$ where $A,B \in \{S,R,D\}$. Also the total power of each channel, \textit{i.e.}, $\mathbb{E}\left[ \vert \vec{h}^{AB} \vert^{2}  \right] = d_{AB}^{-\beta}:= \hat{t}^{2}_{AB} $  }
Let $\alpha$ and $\theta_{n}$ represent amplitude coefficient and phase shift introduced by the $n$-{th} $\mathbf{IRS}$ element, respectively. The signal received at node $\mathbf{D}$ is then given by, 
\textcolor{blue}{$y = \sqrt{p}\left({h}^{SD} +\alpha \left(\vec{h}^{SR}\right)^T \mathbf{\Theta} \vec{h}^{RD}\right) s + w,$} where \textcolor{blue}{$ \mathbf{\Theta} = \ \text{diag} \left( e^{j\theta_{1}},...,e^{j\theta_{N}} \right) $,} $p$ is the transmit power, $s$ is the transmitted signal with $\mathbb{E}$[$|s|^2$]=1 and $w$ is the AWGN with noise power $\sigma^2$.
The SNR at the node $\mathbf{D}$  of the IRS-supported
network is then given by
\textcolor{blue}{$\gamma_{IRS} = \gamma_s \Big\vert {h}^{SD} + \alpha\left(\vec{h}^{SR}\right)^T \mathbf{\Theta} \vec{h}^{RD}\Big\vert^{2},$}
where $\gamma_s=\frac{p}{\sigma^2}$. \textcolor{blue}{To achieve maximum SNR at $\mathbf{D}$, the phase-shift of the $n$-th  IRS element needs to be selected as follows,}
\begin{equation}
\theta_{n}^{opt} = \text{arg}\left({h}^{SD}\right)- \text{arg}\left(\left[\vec{h}^{SR}\right]_{n}\left[\vec{h}^{RD}\right]_{n}\right).
\label{theta}
\end{equation}
\textcolor{blue}{Note that such a choice of phase shifts that maximise the SNR is very common in literature and is also used by the authors of \cite{li2020ergodic,de2021large,bjornson2019intelligent}.} \cbend 
Let $b$ be the number of bits used to represent the phase. Then the set of all possible phase shifts at each of the IRS element is given by $\lbrace 0 , \frac{2\pi}{2^b} , \cdots, \frac{(2^b-1)2\pi}{2^b} \rbrace$ \cite{li2020ergodic}. Hence, $\theta_{n}^{opt}$ may not be always available and the exact phases shift at the $n$-th IRS element can be represented as $\theta_{n}$ = $\theta_{n}^{opt}$ + $\Phi_n$, where $\Phi_{n}$ denotes the phase error at the $n^{th}$ reflector. Note that, {$-2^{-b}\pi \le \Phi_{n} \le  2^{-b}\pi$} and we model $\Phi_n \sim \mathcal{U}[-2^{-b}\pi,2^{-b}\pi]$, similar to the authors of \cite{badiu2019communication}, \cite{li2020ergodic}. Here, $\mathcal{U}[a,b]$ represents the uniform distribution over the support $[a,b]$. Thus, the expression for SNR incorporating the phase error term is given as follows:
\begin{equation}
    \gamma_{IRS} =\gamma_{s} \left(\Big\vert \big|{h}^{SD}\big| + \alpha \sum\limits_{n=1}^{N}\big\vert\left[\vec{h}^{SR}\right]_{n}\big\vert \big\vert\left[\vec{h}^{RD}\right]_{n}\big\vert e^{j\Phi_{n}}\Big\vert^2 \right).
    \label{snr_err}
\end{equation}
In the next section, we derive an exact expression and then three simple approximations for the CDF of $\gamma_{IRS}$ in (\ref{snr_err}) and hence the OP at $\mathbf{D}$. \cbstart \textcolor{blue}{Note that, our system model is similar to \cite{li2020ergodic}, however we consider the more general case of $\kappa-\mu$ fading here. Furthermore, the focus in \cite{li2020ergodic} was the ergodic capacity and hence they do not derive the CDF of $\gamma_{IRS}$.} \cbend 
\section{Outage probability} \label{sec_op}
Outage at a node is the phenomenon of the instantaneous SNR falling below a particular threshold, say $\gamma$. The OP at node $\mathbf{D}$ can be evaluated as,
\begin{equation}
   P_{outage}=\mathbb{P}\left[\gamma_{IRS}<\gamma\right].
   \label{snr_outage}
\end{equation}
\cbstart 
\textcolor{blue}{Note that $\gamma_{IRS}$ is the square of the absolute value of the sum of a $\kappa-\mu$ RV and a sum of i.i.d. double $\kappa-\mu$ RVs \cite{Bhargav2018:ProductKappaMu} each scaled by the exponential of a uniform RV. The PDF of a double $\kappa-\mu$ RV can be expressed in terms of a double infinite summation of terms involving the Meijer G-functions \cite[(7)]{Bhargav2018:ProductKappaMu}. Hence the distribution of the sum of $N$ such scaled double $\kappa-\mu$ RVs has a very complicated expression\cite{wang2020study,atapattu2020reconfigurable}.} \cbend This makes the characterization of the exact distribution of the OP a mathematically intractable task. So far, there were various kinds of approximations for OP (as shown in Table \ref{literature_summary}) proposed in the open literature. \cbstart \textcolor{blue}{However, to the best of our knowledge, none of them considered the most general case, \textit{i.e.}, the presence of an SD link, b-bit quantization error and $\kappa-\mu$ fading channels.} \cbend Hence, in the subsequent subsections, we consider the most general case and present an exact expression and three different simple approximations for (\ref{snr_outage}).
\subsection{Dimension Reduction Approximation
} \label{sec_dim_reduc}
\cbstart 
\color{blue}
In this subsection, we first derive an exact expression for the CDF of SNR in the form of a multi-fold integration where the order of integration grows linearly with the number of elements in the IRS. Solving this multi-fold integration analytically is mathematically intractable, and we demonstrate how the method of uni-variate dimension reduction \cite{Rahman2004:Integral_DimensionReduction} can be used to circumvent this issue.
\begin{lemma}\label{Lemma:LIS_Exact} 
For a threshold $\gamma$, an exact expression for OP at node $\mathbf{D}$ is given by (\ref{Eq:LIS_ExactCDF_Final}), where $\mathbf{c} = \left[ \alpha\cos(\phi_{1}) \dots \alpha\cos(\phi_{N})  \right]^{T} $, $\mathbf{s} = \left[ \alpha\sin(\phi_{1}) \dots \alpha\sin(\phi_{N})  \right]^{T} $, $\mathbf{x} = \left[x_{1},\dots,x_{N}\right]$, $ \rho_{SR} = e^{\mu_{SR}\kappa_{SR}}, \rho_{RD} = e^{\mu_{RD}\kappa_{RD}} $ and $a_{SR} = \frac{\mu_{SR}\left(1+\kappa_{SR}\right) }{\hat{t}^{2}_{SR}}, a_{RD} = \frac{\mu_{RD}\left(1+\kappa_{RD}\right) }{\hat{t}^{2}_{RD}} $. Here, $\textit{K}_{\nu}\left(\cdot\right)$ is modified Bessel function of the second kind of order $\nu$ \cite{BesselKdef}, $Q_{M}\left(a,b\right)$ is the Marcum-Q function and $\operatorname{U}(\cdot)$ is the unit step function.
\begin{align}  
P_{outage} &= \left(\frac{ 2^{b+1} \sqrt{a_{SR} a_{RD}}}{\pi \sqrt{\gamma_{s}}\rho_{SR}\rho_{RD}} \right)^{N} \int\dots\int 
             \left[\scriptstyle{\left(  1- \text{Q}_{\mu_{SD}}\left(\sqrt{2\mu_{SD}\kappa_{SD}},\sqrt{2\mu_{SD}\left(1 + \kappa_{SD}\right)} \frac{\left(\sqrt{\gamma - \left(\mathbf{s}^{T}\mathbf{x}\right)^{2}}- \mathbf{c}^{T}\mathbf{x}\right)}{\sqrt{\gamma_{s}}\hat{t}_{SD}} \right) \right)} \right.  \nonumber
             \\&\hspace{-15mm}
 \times\operatorname{U}\left(\sqrt{\gamma - \left(\mathbf{s}^{T}\mathbf{x}\right)^{2}}- \mathbf{c}^{T}\mathbf{x}\right) \prod_{i=1}^{N} \left( \sum_{m=0}^{\infty} \sum_{n=0}^{\infty} \frac{\left(\mu_{SR}\kappa_{SR} \right)^{m}\left(\mu_{RD}\kappa_{RD} \right)^{n}}{m! \Gamma\left( \mu_{SR} + m\right)n!\Gamma\left( \mu_{RD} + n\right)} \right. \nonumber 
              \\& \left. \left. \hspace{-15mm} \times \left(\sqrt{\frac{a_{SR}a_{RD}}{\gamma_{s}}} x_{i}\right)^{\mu_{SR}+\mu_{RD}+m+n-1}
             \textit{K}_{\mu_{SR}-\mu_{RD} + m-n}\left(2\sqrt{\frac{a_{SR}a_{RD}}{\gamma_{s}}} x_{i} \right) 
             \right) \right] d x_{1} d \phi_{1} \dots d x_{N} d \phi_{N},
 \label{Eq:LIS_ExactCDF_Final}            
\end{align}
\end{lemma}
\color{black}
\cbend 
\begin{proof}
Please refer Appendix \ref{proof:LIS_Exact} for the proof.
\end{proof}
Note that (\ref{Eq:LIS_ExactCDF_Final}) provides an exact expression for the OP at $\mathbf{D}$, but it is a multi-fold integration of order $2N$. It is very difficult to evaluate this expression even numerically for values of $N$ such as $50$ using common mathematical software such as Matlab/Mathematica. So, it is important to have an approximation that is close to (\ref{Eq:LIS_ExactCDF_Final}) and is also easily computable. One such approximation for general multivariate integrals is presented in \cite{Rahman2004:Integral_DimensionReduction}  in the context of stochastic mechanics. The work in  \cite{Rahman2004:Integral_DimensionReduction}  approximate an $n$-th order integration as a  sum of $n$ single order integrations. This approximation has also been recently used to approximate complicated multidimensional integrals for cell-free massive MIMO system\cite{shekhar2021outage}. In this work, we propose to approximate the integral in  (\ref{Eq:LIS_ExactCDF_Final}) using the uni-variate dimension reduction method. The key challenge in applying the dimension reduction method is to express the desired multi-fold integration as an expectation of some function, where the expectation is taken with respect to a random vector of length same as the order of integration. The proposed approximation and detailed proof are presented in the following theorem. 
\color{blue}
\cbstart
\begin{theorem} \label{Thm:LIS_DimReduc}
For a threshold $\gamma$, an approximate expression of the OP at node $\mathbf{D}$ is given by (\ref{Eq:LIS_DimReducApprox}), where $T_{1} = \sqrt{\gamma} - \left(N-1\right)\alpha\mu$, $\mu = \frac{\sqrt{\gamma_{s}}\left( \mu_{SR} \right)_{\frac{1}{2}} \left( \mu_{RD} \right)_{\frac{1}{2}}}{\sqrt{a_{SR} a_{RD}}} {}_{1}F_{1}\left(-\frac{1}{2} ; \mu_{SR} ; -\kappa_{SR} \mu_{SR}\right) {}_{1}F_{1}\left(-\frac{1}{2} ; \mu_{RD} ; -\kappa_{RD} \mu_{RD}\right)$, $h_{1}\left(x\right) = \left(  1- \text{Q}_{\mu_{SD}}\left(\sqrt{2\mu_{SD}\kappa_{SD}},\sqrt{2\mu_{SD}\left(1 + \kappa_{SD}\right)} \frac{\left(\sqrt{\gamma} - x - (N-1)\alpha\mu\right)}{\sqrt{\gamma_{s}}\hat{t}_{SD}} \right) \right)$ and $\textit{K}_{\nu}$ is modified Bessel function of the second kind of order $\nu$ \cite{BesselKdef} $Q_{M}\left(a,b\right)$ is the Marcum-Q function and $\operatorname{U}(\cdot)$ is the unit step function. 
 \end{theorem}
  \begin{proof}
Please refer Appendix \ref{proof:LIS_DimReduc} for the proof. 
\end{proof}
\begin{align}
& P_{outage} \approx N \left( \frac{4\sqrt{a_{SR} a_{RD}}}{\sqrt{\gamma_{s}}\rho_{SR}\rho_{RD}}\right) \sum_{m=0}^{\infty} \sum_{n=0}^{\infty}
         \frac{\left(\mu_{SR}\kappa_{SR} \right)^{m}}{m! \Gamma\left( \mu_{SR} + m\right)}\frac{\left(\mu_{RD}\kappa_{RD} \right)^{n}}{n!\Gamma\left( \mu_{RD} + n\right)} \left(  
         \int_{0}^{T_{1}} h_{1}\left(x\right)  \right.  \nonumber  \\& \times 
 \left(\sqrt{\frac{a_{SR}a_{RD}}{\gamma_{s}}}  x\right)^{\mu_{SR}+\mu_{RD}+m+n-1} \left. K_{\mu_{SR}-\mu_{RD} + m-n}\left(2\sqrt{\frac{a_{SR}a_{RD}}{\gamma_{s}}} x \right) dx  \right) \nonumber
        \\
        &+N \frac{ 2^{b}}{2\pi} \int^{\frac{\pi}{2^{b}}}_{\frac{-\pi}{2^{b}}}
        \begin{array}{l}
            \left(    1- \text{Q}_{\mu_{SD}}\left(\sqrt{2\mu_{SD}\kappa_{SD}},\sqrt{2\mu_{SD}\left(1 + \kappa_{SD}\right)} \frac{\left(\sqrt{\gamma - \left(\alpha\sin{\left( \phi_{i} \right)}\mu \right)^{2}}- 
        \alpha\cos{\left(\phi_{i}\right)} \mu - (N-1)\alpha\mu\right)}{\sqrt{\gamma_{s}}\hat{t}_{SD}} \right) \right) \nonumber \\
         \times\operatorname{U}\Bigg( \Bigg.\sqrt{\gamma - \left(\alpha\sin{\left( \phi \right)}\mu \right)^{2}}-
        \alpha\cos{\left( \phi \right)} \mu - (N-1)\alpha\mu \Bigg.\Bigg)  d \phi   \nonumber 
        \end{array}
         \\
     &-(2N-1)  \left(  1- \text{Q}_{\mu_{SD}}\left(\sqrt{2\mu_{SD}\kappa_{SD}},\sqrt{2\mu_{SD}\left(1 + \kappa_{SD}\right)} \frac{\left(\sqrt{\gamma}- N\alpha\mu\right)}{\sqrt{\gamma_{s}}\hat{t}_{SD}} \right) \right) \operatorname{U}\left(\sqrt{\gamma}- N\alpha\mu \right),
     \label{Eq:LIS_DimReducApprox}
    \end{align}
\color{black}
\cbend 
 The approximation proposed in (\ref{Eq:LIS_DimReducApprox}) consists of only $2$ single integrations\footnote{Note that for practical values of $d_{sr},d_{rd}$ and $\beta$ the value of $\mu$ is quite small and hence, in such regimes  $\gamma - \left(\alpha\sin{\left( \phi \right)}\mu \right)^{2} > 0$ for the range of thresholds typically used. Furthermore, for large $b$ the range of $\sin{\phi}$ itself is small and that too make $\gamma - \left(\alpha\sin{\left( \phi \right)}\mu \right)^{2} $ positive. Hence ${\gamma - \left(\alpha\sin{\left( \phi \right)}\mu \right)^{2}}$ which appears in the second integral is always non negative for all the cases of interest in this application.} which are easy to solve numerically and provides a good approximation for the OP which is confirmed through extensive simulations presented in Section \ref{simulation}. 
 \par Next, we will discuss one case which is prevalent in literature. We consider the scenario where we have a perfect phase alignment at the IRS, \textit{i.e.}, there is no phase error, and the IRS has no hardware impairments. The approximation of the OP, for perfect phase alignment case,  can be obtained by substituting $\phi = 0$ in (\ref{Eq:LIS_DimReducApprox}) and observing that we are approximating an $N$-th order integration. The final result is presented in the following corollary.
\color{blue}
\cbstart
\begin{corollary}\label{Cor:LIS_dimreduc_nophase}
For perfect phase alignment and threshold $\gamma$, the OP at node $\mathbf{D}$ is approximated as
\begin{align}
\label{Eq:LIS_DimReducApproxNoPhaseerror}
    % \begin{aligned}
         P_{outage} &\approx 
        N \left( \frac{4\sqrt{a_{SR} a_{RD}}}{\sqrt{\gamma_{s}}\rho_{SR}\rho_{RD}}\right) \sum_{m=0}^{\infty} \sum_{n=0}^{\infty}
         \frac{\left(\mu_{SR}\kappa_{SR} \right)^{m}}{m! \Gamma\left( \mu_{SR} + m\right)}\frac{\left(\mu_{RD}\kappa_{RD} \right)^{n}}{n!\Gamma\left( \mu_{RD} + n\right)} \left(  
         \int_{0}^{T_{1}} h_{1}\left(x\right)  \right. \nonumber  \\& \times 
            \left(\sqrt{\frac{a_{SR}a_{RD}}{\gamma_{s}}}  x\right)^{\mu_{SR}+\mu_{RD}+m+n-1} \left. K_{\mu_{SR}-\mu_{RD} + m-n}\left(2\sqrt{\frac{a_{SR}a_{RD}}{\gamma_{s}}} x \right) dx  \right)
        \\ 
     &-(N-1)  \left(  1- \text{Q}_{\mu_{SD}}\left(\sqrt{2\mu_{SD}\kappa_{SD}},\sqrt{2\mu_{SD}\left(1 + \kappa_{SD}\right)} \frac{\left(\sqrt{\gamma}- N\alpha\mu\right)}{\sqrt{\gamma_{s}}\hat{t}_{SD}} \right) \right) \operatorname{U}\left(\sqrt{\gamma}- N \alpha\mu \right). \nonumber
    % \end{aligned}
\end{align}
\end{corollary} 
\color{black}
\cbend
\par While the uni-variate approximation is highly accurate, we could not apply it to the case of no SD link. In order to address this gap, we look at moment matching in the following subsection.
\subsection{Gamma approximation using Moment matching}  \label{sec_mom}
In this sub-section, we approximate the SNR as a Gamma RV with shape parameter $k_{mom}$ and scale parameter $\theta_{mom}$ by matching their first and second moments. Using this result, the OP at node $\mathbf{D}$ is given by the following theorem.
\begin{theorem} \label{gamma_approx}
The OP for a threshold $\gamma$ at node $\mathbf{D}$ can be evaluated as
\begin{equation}
P_{outage}=\frac{\gamma^{k_{mom}}}{\theta_{mom}^{k_{mom}} \Gamma\left(k_{mom}+1\right)}{ }_{1} F_{1}\left(k_{mom}, k_{mom}+1, \frac{-\gamma}{\theta_{mom}}\right),
\label{p_out_mom}
\end{equation}
where the shape parameter ($k_{mom}$) and the scale parameter ($\theta_{mom}$) of the Gamma distribution can be evaluated using:
\begin{equation}
    \theta_{mom}= \frac{\mathbb{E}[\gamma_{IRS}^2] - \mathbb{E}^2[\gamma_{IRS}]}{E[\gamma_{IRS}]},
    \label{theta}
\end{equation}
\begin{equation}
    k_{mom} =\frac{\mathbb{E}[\gamma_{IRS}]}{\theta_{mom}}.
    \label{k}
\end{equation}
Here, ${ }_{1}F_{1}(\cdot,\cdot,\cdot)$ is the confluent hypergeometric function of the first kind \cite{confluent} and $\mathbb{E}[\gamma_{IRS}], \mathbb{E}[\gamma_{IRS}^2]$ can be evaluated using (\ref{mean_1_general_fad}) and  (\ref{mean_2_general_fad}) respectively. 
\end{theorem}
\begin{proof}
Please refer Appendix \ref{gamma} for the proof.
\end{proof}
Note that the expression in (\ref{p_out_mom}) is very easy to evaluate when compared to the OP approximations proposed in a few of the recent literature including \cite{wang2020study,wang2020chernoff,zhang2019analysis}. Also, the proposed approximations hold well both for the cases of small and large values of $N$, unlike the Gaussian approximations using CLT \cite{tang2020physical,9369134,yang2020secrecy,xu2020ergodic,kudathanthirige2020performance} which holds only for large $N$. \cbstart \textcolor{blue}{Furthermore, the expressions for the first and second moments presented in (\ref{mean_1_general_fad}) and (\ref{mean_2_general_fad}) are general and can be used to evaluate the statistics of SNR for any fading scenario once we have the first four moments of the underlying fading distribution, which are readily available for most of the well-known fading models.} \cbend Also, the proposed CDF of $\gamma_{IRS}$ can be easily used for deriving the expressions of other metrics of interest like the rate \cite{srinivasan2017approximate}. Since we have considered a very general scenario in Theorem \ref{gamma_approx}, we present certain special cases of interest in the following corollaries.  
\color{blue}
\cbstart
\begin{corollary} \label{without_phase_error}
In the absence of phase errors, the OP can be approximated using (\ref{p_out_mom}), where 
(\ref{theta}) and (\ref{k}) can be evaluated using the following expressions for the moments of SNR.
\begin{equation}
    \mathbb{E}[\gamma_{IRS}^{np}] = \gamma_{s}\left( m_2^{SD}+N\alpha^2 m_2^{SR} m_2^{RD}+2N\ \alpha m_1^{SD} m_1^{SR} m_1^{RD}+N(N-1) \alpha^2(m_1^{SR})^2 (m_1^{RD})^2 \right).
    \label{meanK_no_phase_error}
\end{equation}
\begin{align} \label{mean_K_2_no_phase_error}
 &\mathbb{E}[(\gamma^{np}_{IRS})^2] =  \gamma^2_{s}\left\lbrace m_4^{SD}+2\alpha^2N m_2^{SR}m_2^{RD}m_2^{SD}+N \alpha^4 m_4^{SR}m_4^{RD}+N(N-1) \alpha^4 (m_2^{SR})^2 (m_2^{RD})^2 \right. \nonumber \\
 & \left.+ 4\alpha^2 m_2^{SD}\left[N m_2^{SR}m_2^{RD} +N(N-1)(m_2^{SR})^2 (m_2^{RD})^2 \right] + N(N-1)\alpha^4 \left[2 (N-2) (1+p) \right.  \right. \nonumber \\ 
 &\left. \left. m_2^{SR}m_2^{RD}(m_1^{SR})^2(m_1^{RD})^2 +(m_2^{SR})^2(m_2^{RD})^2(1+p^2)+(N-2)(N-3)(m_1^{SR})^4(m_1^{RD})^4\right] \right. \nonumber \\ 
 &\left.+ 4N\alpha \left[ m_3^{SD} m_1^{SR} m_1^{RD} + \alpha^2  m_1^{SD} m_3^{SR} m_3^{RD} +(N-1) \alpha^2 m_1^{SD} m_1^{SR} m_1^{RD} m_2^{SR} m_2^{RD}
 \right] \right. \nonumber \\ 
 &\left.+ 4N(N-1)\alpha^3 \left[ m_1^{SD} \left(m_1^{SR} \right)^3\left(m_1^{RD}  \right)^3
(N-2)  +(1+p) m_1^{SD}  m_2^{SR}m_2^{RD} m_1^{SR} m_1^{RD} 
\right]\right. \nonumber \\ 
 &\left.+ 2 \alpha^2 N(N-1)  m_1^{SR} m_1^{RD}  \left[m_2^{SD} m_1^{SR} m_1^{RD} +\alpha^2 \left(2 m_3^{SR} m_3^{RD}  + m_1^{SR}  m_1^{RD} m_2^{SR} m_2^{RD} (N-2)\right)
 \right]\right\rbrace,
\end{align}
where $m_p^{AB}:=\mathbb{E}\left[|[\vec{h}^{AB}]_n|^p\right]$ for $p \in \{1,2,3,4\}$ and $A,B \in \{S,R,D\}$.
\end{corollary}
\begin{proof}
Equations (\ref{meanK_no_phase_error}) and (\ref{mean_K_2_no_phase_error}) are obtained from (\ref{mean_1_general_fad}) and (\ref{mean_2_general_fad}) by substituting $b \rightarrow \infty$.
\end{proof}
\begin{corollary} 
\label{without_sd}
When the SD link is in a permanent outage, the OP can be approximated using (\ref{p_out_mom}) where 
equations (\ref{theta}) and (\ref{k}) can be evaluated using the following expressions:
\begin{equation}
    \mathbb{E}[\gamma_{IRS}^{ndl}] = \gamma_{s}N \alpha^2 \left(  m_2^{SR} m_2^{RD}+(N-1) (m_1^{SR})^2 (m_1^{RD})^2 s^2\right).
    \label{meanK_no_sd}
\end{equation}
\begin{equation}
    \begin{aligned}
      &  \mathbb{E}[\left(\gamma_{IRS}^{ndl}\right)^2] = N \alpha^4 \left \lbrace m_4^{SR}m_4^{RD}+(N-1) (m_2^{SR})^2 (m_2^{RD})^2+(N-1) \left[ 2(N-2)m_2^{SR}m_2^{RD} \right. \right.\\ & \left. \left.  s^2\left(1+p\right)(m_1^{SR})^2(m_1^{RD})^2+(m_2^{SR})^2(m_2^{RD})^2(1+p^2)+s^4(N-2)(N-3)(m_1^{SR})^4(m_1^{RD})^4\right]+\right.\\& \left. 4(N-1)  m_3^{SR} m_3^{RD} m_1^{SR} m_1^{RD} s^2+2(N-1)(N-2)  (m_1^{SR})^2  (m_1^{RD})^2 m_2^{SR} m_2^{RD} s^2 \right\rbrace.
    \label{mean_K_2_no_sd}
    \end{aligned}
\end{equation}
where $s=\frac{2^b}{\pi}\sin\left(\frac{\pi}{2^b}\right), p=\frac{2^b}{2\pi }\sin\left(\frac{2\pi}{2^b}\right)$.
\end{corollary} 
\cbend 
\color{black}
\begin{proof}
When the SD link is in a permanent outage, the phase error at the $n$-th reflector element is given by $\Phi_n:=\theta_n-\theta_{opt}$
where $\theta_{opt}=-\text{arg}\left(\left[\vec{h}^{SR}\right]_{n}\left[\vec{h}^{RD}\right]_{n}\right)$. Here also, we model the phase error as a uniform RV, \textit{i.e.}, $\Phi_n \sim \mathcal{U}\left[-2^{-b}\pi,2^{-b}\pi\right]$. In this case, the SNR expression given in equation (\ref{snr_err}) can be modified as follows:
\begin{equation}
    \gamma_{IRS}^{ndl} =\gamma_{s} \Big| \alpha \sum\limits_{n=1}^{N}\big\vert\left[\vec{h}^{SR}\right]_{n}\big\vert \big\vert\left[\vec{h}^{RD}\right]_{n}\big\vert e^{j\Phi_{n}}\Big|^2.
    \label{snr_err_ndl}
\end{equation}
Next, we follow the steps similar to Appendix \ref{gamma} and arrive at (\ref{meanK_no_sd}) and (\ref{mean_K_2_no_sd}).
\end{proof}
\textcolor{blue}{Note that Corollary \ref{without_sd} recover existing results presented in  \cite[(9)]{de2021large} and \cite[(10)]{de2021large} for the values of  $\kappa_{SD}=\kappa_{SR}=\kappa_{RD}=0$ and $\mu_{SD}=\mu_{SR}=\mu_{RD}=1$, (i.e all links are Rayleigh channels).}
\color{blue}
\cbstart 
\begin{corollary}
\label{without_sd_without_phase_corollary}
In the absence of phase errors and when the SD link is in a permanent outage, the OP can be approximated using (\ref{p_out_mom}) where 
equations (\ref{theta}) and (\ref{k}) can be evaluated using the following expressions:
\begin{equation}
    \mathbb{E}[\gamma_{IRS}^{npdl}] = \gamma_{s}N \alpha^2 \left(  m_2^{SR} m_2^{RD}+(N-1) (m_1^{SR})^2 (m_1^{RD})^2 \right).
    \label{meanK_no_phase_error_no_sd_link}
\end{equation}
\begin{align}
& \mathbb{E}[(\gamma^{npdl}_{IRS})^2] =  N \alpha^4 \left\lbrace  m_4^{SR}m_4^{RD}+(N-1)  (m_2^{SR})^2 (m_2^{RD})^2+(N-1) \left[ 4(N-2)m_2^{SR}m_2^{RD} \right. \right. \nonumber \\ & \left. \left. (m_1^{SR})^2(m_1^{RD})^2+2(m_2^{SR})^2(m_2^{RD})^2+(N-2)(N-3)(m_1^{SR})^4(m_1^{RD})^4\right]+4(N-1)  m_3^{SR} m_3^{RD} \right.  \nonumber \\& \left.  m_1^{SR} m_1^{RD} +2(N-1)(N-2)  (m_1^{SR})^2  (m_1^{RD})^2 m_2^{SR} m_2^{RD} \right\rbrace.
  \label{mean_K_2_no_phase_error_no_sd_link}
\end{align}
\end{corollary}
\color{black}
\cbend 
\begin{proof}
Equations (\ref{meanK_no_phase_error_no_sd_link}) and (\ref{mean_K_2_no_phase_error_no_sd_link}) are obtained from (\ref{meanK_no_sd}) and (\ref{mean_K_2_no_sd}) by substituting 
$s,p \rightarrow 1$.
\end{proof}
\textcolor{blue}{Note that for the values of $\kappa_{SD}=\kappa_{SR}=\kappa_{RD}=0$ and $\mu_{SD}=\mu_{SR}=\mu_{RD}=1$, (i.e all links are Rayleigh channels), Corollary  \ref{without_sd_without_phase_corollary} recover existing results presented in  \cite[(14)]{de2021large} and \cite[(15)]{de2021large}, respectively.}
\cbend 
\par Although the uni-variate approximation is close to the simulated OP, but it's not possible for us to derive it for all possible cases considered in this paper. The gamma approximation based on moment matching is simple and could be applied to all considered scenarios but was not as tight as the uni-variate approximation. In the next sub-section, we look at another gamma approximation which has the minimum KL divergence with respect to the exact distribution of SNR.
\subsection{KL divergence minimization
} \label{sec_kl_div_min}
In this section, we identify the parameters of a Gamma distribution such that the KL divergence between the resulting RV and the exact SNR is the least among all possible Gamma distributions. Using this result, the OP at $\mathbf{D}$ is given by the following theorem:
\begin{theorem} \label{thm_kl_div_min}
The OP for a threshold $\gamma$ at the node $\mathbf{D}$ is given by 
\begin{equation}
P_{outage}=\frac{\gamma^{k_{kl}}}{\theta_{kl}^{k_{kl}} \Gamma\left(k_{kl}+1\right)}{ }_{1} F_{1}\left(k_{kl}, k_{kl}+1, \frac{-\gamma}{\theta_{kl}}\right),
\label{p_out_kl}
\end{equation}
where $k_{kl}$ and $\theta_{kl}$ are obtained by solving the following two equations:
\begin{align}
\mathbb{E}[\log(\gamma_{IRS})] & = \log(\theta_{kl}) + \psi(k_{kl}), \label{log_e_11} \\
\mathbb{E}[\gamma_{IRS}] &= k_{kl} \times \theta_{kl}.  \label{e_11}
\end{align}
Here, ${ }_{1}F_{1}(.,.,)$ is the confluent hypergeometric function of the first kind \cite{confluent} and $\psi(.)$ is the digamma function \cite{digamma}.
\end{theorem}
\begin{proof}
Please refer Appendix \ref{proof_kl_min} for the proof. 
\end{proof}
The derivation of the exact expression for evaluating $\mathbb{E}[\log(\gamma_{IRS})]$ is complicated and hence we proceed with the following approximation for the same \cite[(11)]{srinivasan2017approximate}:
\begin{equation}
\mathbb{E}[\log \gamma_{IRS}] \approx \log \left(\mathbb{E}[\gamma_{IRS}] \right)-\frac{1}{2} \frac{\mathbb{E}[\gamma_{IRS}^2]-\mathbb{E}^2[\gamma_{IRS}] }{\mathbb{E}^2[\gamma_{IRS}]}.
\label{approx_mean_log}
\end{equation}
Given that we can compute the first and second moments of $\gamma_{IRS}$ using (\ref{mean_1_general_fad}) and (\ref{mean_2_general_fad}), we can easily evaluate (\ref{approx_mean_log}). Then we can solve for the parameters $k_{kl}$ and $\theta_{kl}$ using the solvers available in any mathematical software such as \textit{Matlab, Mathematica, or Octave}. Thus, the method of KL divergence minimization also provides us with a simple expression for the OP that is very amenable for computation and further analysis.
\begin{corollary}
For the special cases without SD link or phase error or both, the OP for a threshold $\gamma$ is given by (\ref{p_out_kl}). Corresponding values of scale and shape parameters can be solved using equations (\ref{log_e_11}) and (\ref{e_11}), and the corresponding moments can be evaluated using \ref{meanK_no_phase_error}-\ref{mean_K_2_no_phase_error_no_sd_link}.  
\end{corollary}
Note that the approximation for the expectation of the logarithm of SNR provided in (\ref{approx_mean_log}) makes use of only the first and second moment of $\gamma_{IRS}$ and this approximation does not hold equally well throughout the support of $\gamma_{IRS}$. This was particularly observed in the simulations of certain special cases like scenarios without the SD link. Hence, we propose the following method to circumvent this issue for scenarios without an SD link and no phase error. In such a case we have, 
\begin{equation}
    \gamma_{IRS}^{npdl} = \gamma_{s} \left(  \alpha \sum\limits_{n=1}^{N}\big\vert\left[\vec{h}^{SR}\right]_{n}\big\vert\big\vert\left[\vec{h}^{RD}\right]_{n}\big\vert \right)^2.
\end{equation}
\cbstart
\textcolor{blue}{Now, we can approximate the double $\kappa-\mu$ RV $\left[\gamma_{dr}\right]_{n} :=\big\vert\left[\vec{h}^{SR}\right]_{n}\big\vert\big\vert\left[\vec{h}^{RD}\right]_{n}\big\vert$ as a Gamma RV with shape parameter and scale parameter $k_{kl,dr}$ and $\theta_{kl,dr}$ respectively using the method of KL divergence minimisation\footnote{The authors of \cite{wang2020study} also approximates double Rayleigh RVs as a Gamma RV, but using the method of moment matching.}. In this case, the expectation of the logarithm of the RV $\left[\gamma_{dr}\right]_{n}$ has a closed-form expression (as given in \ref{mean_log_dr}) and hence we can avoid the approximation used in (\ref{approx_mean_log}). } \cbend 
\begin{equation}
\begin{aligned}
& \mathbb{E}\left[\log\left(\left[\gamma_{dr}\right]_{n}\right) \right] = \frac{-\kappa_{SR}^{\frac{1}{2}-\frac{\mu_{SR}}{2}}}{2} (1+\kappa_{SR})^{\frac{1+\mu_{SR}}{2}} \mu_{SR}\left(\frac{(1+\kappa_{SR}) \mu_{SR}}{\Omega_{sr}}\right)^{-\mu_{SR}}\left(\frac{(\kappa_{SR}^2+\kappa_{SR}) \mu_{SR}^{2}}{\Omega_{sr}}\right)^{\frac{-1+\mu_{SR}}{2}} \\
&\Omega_{sr}^{\frac{-1-\mu_{SR}}{2}}\left(\log \left(\frac{(1+\kappa_{SR}) \mu_{SR}}{\Omega_{sr}}\right)-\psi (\mu_{SR})+{}_1F_1^{(1,0,0)}(0, \mu_{SR},-\kappa_{SR} \mu_{SR})\right)-\frac{\kappa_{RD}^{\frac{1-\mu_{RD}}{2}}}{2} (1+\kappa_{RD})^{\frac{1+\mu_{RD}}{2}} \\& \mu_{RD} \left(\frac{(1+\kappa_{RD}) \mu_{RD}}{\Omega_{rd}}\right)^{-\mu_{RD}}\left(\frac{\kappa_{RD}(1+\kappa_{RD}) \mu_{RD}^{2}}{\Omega_{rd}}\right)^{\frac{-1+\mu_{RD}}{2}} 
\Omega_{rd}^{\frac{-1-\mu_{RD}}{2}}\left(\log \left(\frac{(1+\kappa_{RD}) \mu_{RD}}{\Omega_{rd}}\right) \right.\\ & \left. -\psi (\mu_{RD}) +{}_1F_1^{(1,0,0)}(0, \mu_{RD},-\kappa_{RD} \mu_{RD})\right), \ \forall \ n,
\end{aligned}
\label{mean_log_dr}
\end{equation}
where $\psi(.)$ is the digamma function \cite{Polygamma} and ${}_1F_1^{(1,0,0)}(a,b,z)$ is the derivative of the confluent hypergeometric function with respect to the parameter $a$ and can be evaluated using \cite{diffconfluent}. For the cases with SD link, or $b>1$, we could not arrive at simple expressions for the distribution of $\gamma_{IRS}$ even after approximating the double $\kappa-\mu$ RV as a gamma RV \cite{wang2020study}.  
\subsection{Which approximation should one finally use?}
 The major difference between the approximation proposed in section \ref{sec_dim_reduc} and the approximations proposed in Section \ref{sec_mom} and \ref{sec_kl_div_min} lies in the fact that the later two approximations are obtained by matching certain moments/statistics of $\gamma_{IRS}$ with the moments/statistics of a specific distribution (Gamma distribution in this case). However, in Section \ref{sec_dim_reduc} we directly approximate the exact integral expression for the OP using the uni-variate dimension reduction method \textcolor{blue}{(except for the scenarios without an SD link).} While the moment matching based approximation may fail to follow the tail of the SNR accurately (this failure to accurately model the tail of the SNR is observed in the next section), the approximation using the uni-variate dimension reduction method is an approximation of the exact CDF and hence, is better at approximating the tail. Moment matching also has its own advantages in terms of having a simple expression. \cbstart \textcolor{blue}{This simple expression is very amenable for further analysis and for the derivation of the statistics of functions of the SNR like the rate, secrecy capacity etc. Furthermore, the closed form expression for OP is highly useful for formulating resource allocation problems where we need certain performance guarantees in terms of the OP. Similarly, the gamma approximation using the KL divergence minimization also results in a simple expression for OP, however the approximation for the expectation of logarithm of the SNR makes this approximation less accurate in certain regimes. However, for the cases without an SD link, we can circumvent this issue as discussed in the last section.} Hence, depending upon the specific application, one can choose among the proposed approximations and also the existing approximations. For example, if one's focus is on the tail behavior of OP then one may choose to use the uni-variate approximation. \cbend 
\section{simulation Results}
\label{simulation}
In this section, we present the observations of simulation experiments to verify the results presented in Section \ref{sec_op}. 
The simulation setting used in this paper is similar to
\cite{li2020ergodic} (shown in Fig \ref{set_up}). Here, nodes $\mathbf{S}$ and $\mathbf{D}$ are located at the points $(0,0)$ and $(90,0)$ respectively and the $\mathbf{IRS}$ is located at the point $(d,h)$. Throughout the simulations, we have taken the value of amplitude coefficient $\alpha$ to be 1 (similar to \cite{li2020ergodic}), $\beta$ is chosen as 4, $b$ is chosen as 5, and $\gamma_s$ to be $73$ dB, unless mentioned otherwise. \textcolor{blue}{Here, we have chosen, $\kappa_{SD}=0.5$, $\mu_{SD}=0.8$, $\kappa_{SR}=1.41$, $\mu_{SR}=2$, $\kappa_{RD}=1.52$, and $\mu_{RD}=$2.5 and $h$ is chosen to be $10$ metres whereas $d$ is varied across simulations\footnote{\textcolor{blue}{We have verified the simulations for different values of $\kappa$ and $\mu$ and observed a good match between the simulated values of OP and the proposed approximations. Due to spave constraints, the results for special cases like the Rayleigh, Rician and Nakagami-m are included in the supplementary document.}}. In the subsequent subsections, we compare the performance of the different approximations proposed for the cases with and without SD link}. 
\cbstart
\color{blue}
\begin{table}[ht]
\begin{minipage}[b]{0.4\linewidth}
\centering
\begin{tabular}{ | l | r | r | r |}
    \hline
    Method & $N=5$ & $N=100$ \\ \hline \hline
    Uni-variate & $6.5536\times10^{-4}$ & $7.0463\times10^{-4}$ \\ \hline
    Moment matching & 0.0140 &0.0170  \\ \hline
    KL divergence & 0.0431  & 0.0295 \\ \hline
   \end{tabular}
    \caption{Table of KS statistic for $d=30$ and $b=5$}
    \label{table_ks}
\end{minipage}\hfill
\begin{minipage}[b]{0.56\linewidth}
 \centering
    \begin{tikzpicture}[scale=0.8]
      \draw[-] (0, 0) -- (9, 0) ;
      \draw[-] (4, 2) -- (4, 0) ;
      \draw[-] (0, 0) -- (4, 2) ;
      \draw[-] (4, 2) -- (9, 0) ;
      \node at(4.5,1) {$\mathbf{h}$};
      \node at(2,-0.5) {$d$};
       \node at(7,-0.5) {90 - $d$};
       \node at(2,1.4) {$d_{sr}$};
       \node at(6.5,1.3) {$d_{rd}$};
      \node at (0, 0) [circle,fill,scale=0.5]{};
     \node at (9, 0) [circle,fill,scale=0.5]{};
      \node at (4, 2) [circle,fill,scale=0.5]{};
    \node at (0, -0.5) []{($0,0$)};
    \node at (0, 0.5) []{$\mathbf{S}$};
    \node at (9, -0.5) []{($90,0$)};
    \node at (9, 0.5) []{$\mathbf{D}$};
    \node at (4, 2.5) []{($d,h$)};
     \node at (4, 3) []{IRS};
      \end{tikzpicture}
     
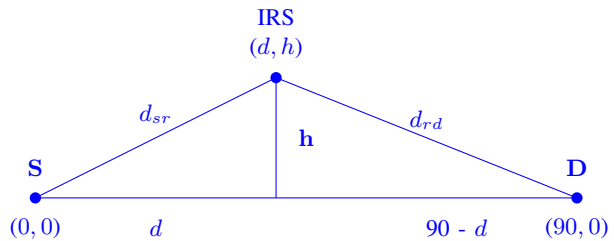
\captionof{figure}{Simulation set up}
     \label{set_up}
\end{minipage}
\end{table}

 \color{black}
\subsection{Results with SD link}
\label{with_sd_results}
\begin{table}
\begin{tabular}{|l||*{8}{c|}}\hline
\multirow{5}{*}{}
$N$ &\backslashbox{Method}{Threshold($\gamma$)}
&\makebox[3em]{-12 dB}&\makebox[3em]{-7 dB}&\makebox[3em]{-5 dB}&\makebox[3em]{-2 dB}
&\makebox[3em]{0 dB}&\makebox[3em]{2 dB}&\makebox[3em]{5 dB}\\\hline \hline
$N=5$ &Simulated &0.2161   &0.4899  &0.6417 &0.8620 &0.9556 &0.9931 &1.0000\\\cline{2-9}
&Uni-variate Approx.  & \textbf{0.2164}  & \textbf{0.4898}  & \textbf{0.6417}  & \textbf{0.8618} & \textbf{0.9558}  & \textbf{0.9931}  & \textbf{1.0000} \\\cline{2-9}
&Gamma($k_{mom},\theta_{mom}$) &0.2037 &0.4884 &0.6454 &0.8655   &0.9561 &  0.9925  &0.9999\\\cline{2-9}
&Gamma($k_{kl},\theta_{kl}$) &   0.1739   & 0.4677    &0.6371    &0.8724  &  0.9629  &  0.9948 &    1.0000\\\hline \hline
$N=50$ &Simulated &0.1530  &0.4185  & 0.5784 &0.8253 &0.9399  &0.9898  &0.9999 \\\cline{2-9}
&Uni-variate Approx. & \textbf{0.1530}  & \textbf{0.4189}  & \textbf{0.5782}  & \textbf{0.8251} & \textbf{0.9399}  & \textbf{0.9897}  & \textbf{0.9999} \\\cline{2-9}
&Gamma($k_{mom},\theta_{mom}$) &0.1475  &0.4137  & 0.5779 &0.8284 &0.9410  &0.9892  &0.9999 \\\cline{2-9}
&Gamma($k_{kl},\theta_{kl}$) &0.1237   & 0.3923   & 0.5668  &  0.8338  &  0.9479  &  0.9920  &  1.0000 \\\hline \hline
$N=100$ &Simulated &0.0909  &0.3390  & 0.5031 &0.7772 &0.9174  & 0.9843 &0.9999 \\\cline{2-9}
&Uni-variate Approx.  & \textbf{0.0908}  & \textbf{0.3394} & \textbf{0.5030} & \textbf{0.7771} & \textbf{0.9174}  & \textbf{0.9843} & \textbf{0.9999} \\\cline{2-9}
&Gamma($k_{mom},\theta_{mom}$) &0.0972  &0.3321  &0.4979  &0.7789 &0.9190   &0.9842  &0.9998  \\\cline{2-9}
&Gamma($k_{kl},\theta_{kl}$) &    0.0800   & 0.3112   &  0.4845   & 0.7821   & 0.9259    &0.9874    &0.9999 \\\hline \hline
\end{tabular}
\caption{Comparison of OP with SD link with varying N when $d=30$, $b=5$}
\label{diff_N}
\end{table}
\begin{figure}[h!]
\begin{minipage}[b]{0.45\linewidth}
\centering
\includegraphics[width=\textwidth]{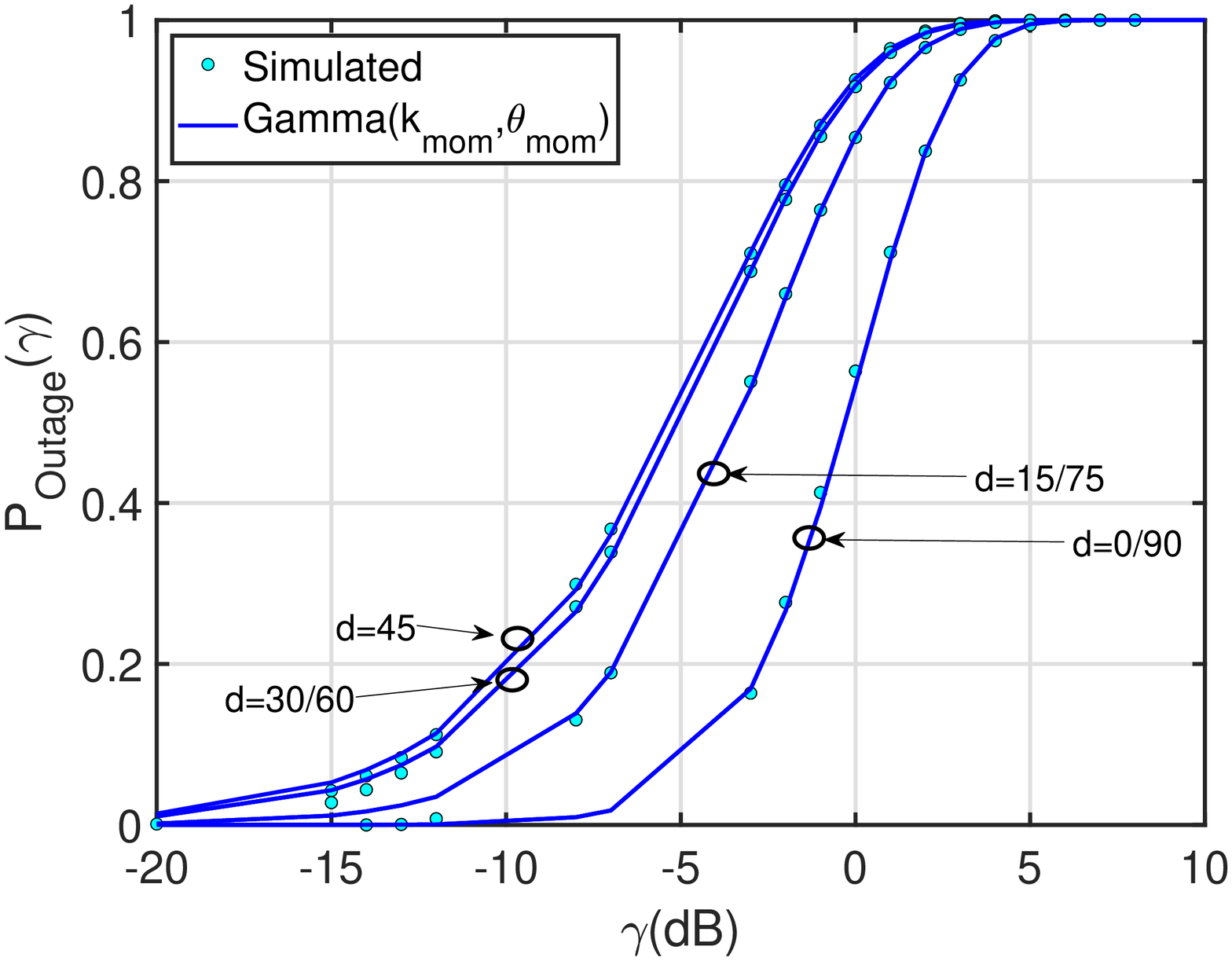}
\caption{Impact of $d$ on the OP with SD link when $N=100$ and b=5}
\label{d_vary_N_150_s_1}
\end{minipage}
\hspace{0.5cm}
\begin{minipage}[b]{0.45\linewidth}
\centering
\includegraphics[width=\textwidth]{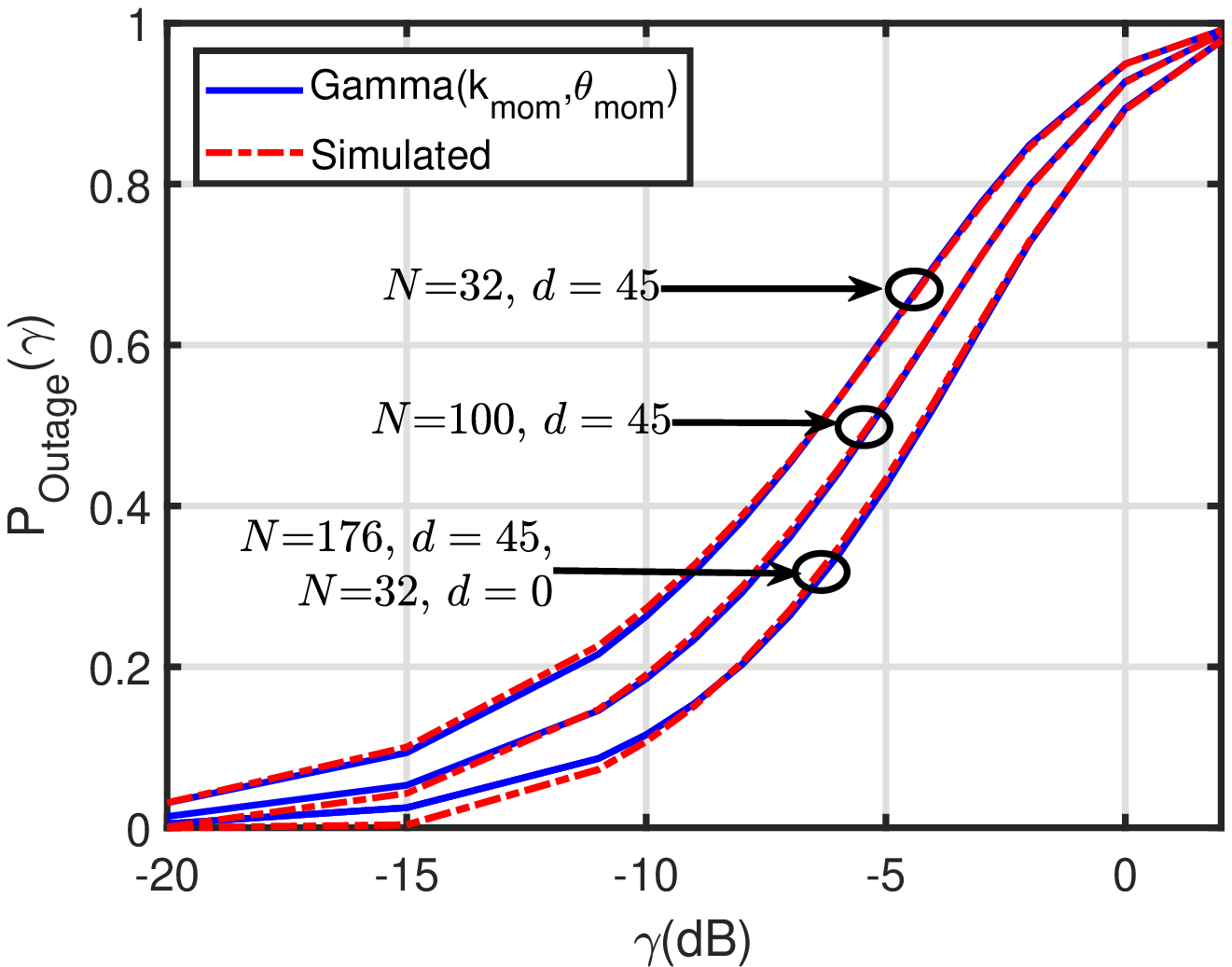}
\caption{Impact of $N$ on the OP with SD link when $b=5$}
\label{vary_N_s_2}
\end{minipage}
\end{figure}
\textcolor{blue}{
In this sub-section, we examine the closeness between the simulated values of OP and the approximations proposed in Section \ref{sec_dim_reduc}-\ref{sec_kl_div_min}. 
Table \ref{diff_N} compares the OP obtained using (\ref{Eq:LIS_DimReducApprox}), (\ref{p_out_mom}), and (\ref{p_out_kl}) with the simulated values of OP for various values of $N$\footnote{We could observe numerical issues in the OP evaluated using (7) for very large values of $\kappa_{AB}$ and $\mu_{AB}$. But in a practical scenario, the values of $\kappa_{AB}$ and $\mu_{AB}$ is rarely larger than $5$. Thus, the proposed results are useful in all scenarios of practical interest.}. Table \ref{diff_N} corroborates the fact that an increase in $N$ improves the performance of an IRS system. One can also observe that the OP evaluated using all the three approximations are close to the simulated values and the uni-variate approximation demonstrates the best performance amongst the three.}
\begin{table}
\begin{tabular}{|l||*{8}{c|}}\hline
\multirow{5}{*}{}
$b$ &\backslashbox{Method}{Threshold($\gamma$)}
&\makebox[3em]{-12 dB}&\makebox[3em]{-7 dB}&\makebox[3em]{-5 dB}&\makebox[3em]{-2 dB}
&\makebox[3em]{0 dB}&\makebox[3em]{2 dB}&\makebox[3em]{5 dB}\\\hline \hline 
$b=1$ &Simulated &0.1348 &0.3964  & 0.5578 &0.8127  &0.9342  &0.9885 &0.9999 \\  \cline{2-9}
&Uni-variate Approx.  & 0.1319  & \textbf{0.3968}  & \textbf{0.5588}  & \textbf{0.8146} & \textbf{0.9358}  & \textbf{0.9889}  & \textbf{0.9999} \\\cline{2-9}
&Gamma($k_{mom},\theta_{mom}$) &0.1322  &  0.3908  &0.5562  &0.8156  &0.9355  & 0.9880 &0.9999 \\ \cline{2-9}
&Gamma($k_{kl},\theta_{kl}$) & 0.1103   & 0.3694   & 0.5444    &0.8204   & 0.9424   & 0.9909   & 0.9999
\\\hline \hline
$b=2$ &Simulated &0.1021 & 0.3545 & 0.5181 &0.7872 & 0.9223  &0.9856  &0.9999  \\ \cline{2-9}
&Uni-variate Approx.  & \textbf{0.1019}  & \textbf{0.3550}  & \textbf{0.5181}  & \textbf{0.7873} & \textbf{0.9224}  & \textbf{0.9855}  & \textbf{0.9999} \\\cline{2-9}
&Gamma($k_{mom},\theta_{mom}$) &0.1060 &  0.3478 &0.5139  &0.7893 &0.9238  &0.9853 &0.9998 \\ \cline{2-9}
&Gamma($k_{kl},\theta_{kl}$) &  0.0876   & 0.3267   & 0.5008  &  0.7930   & 0.9307  &  0.9885   & 0.9999\\\hline \hline
$b=10$ &Simulated    &0.0907  & 0.3388 & 0.5028   & 0.7770  &0.9174 &0.9843   &0.9999 \\ \cline{2-9}
&Uni-variate Approx.  & \textbf{0.0906}  & \textbf{0.3392}  & \textbf{0.5027}  & \textbf{0.7770} & \textbf{0.9174}  & \textbf{0.9843}  & \textbf{0.9999} \\\cline{2-9}
&Gamma($k_{mom},\theta_{mom}$) & 0.0970    & 0.3318 &0.4977  &0.7787   &0.9189  &0.9842  &0.9998\\ \cline{2-9}
&Gamma($k_{kl},\theta_{kl}$)& 0.0799   & 0.3110  &  0.4842   & 0.7819   & 0.9258   & 0.9874   & 0.9999\\\hline \hline
$b=\infty$ &Simulated &0.0907  & 0.3388 & 0.5028   & 0.7770  &0.9174 &0.9843   &0.9999 \\ \cline{2-9}
&Uni-variate Approx.  & \textbf{0.0906}  & \textbf{0.3392}  & \textbf{0.5027}  & \textbf{0.7770} & \textbf{0.9174}  & \textbf{0.9843}  & \textbf{0.9999} \\\cline{2-9}
&Gamma($k_{mom},\theta_{mom}$) & 0.0970    & 0.3318 &0.4977  &0.7787   &0.9189  &0.9842  &0.9998\\ \cline{2-9}
&Gamma($k_{kl},\theta_{kl}$) &  0.0799    &0.3110   & 0.4842   & 0.7819   & 0.9258   & 0.9874   & 0.9999\\\hline \hline
\end{tabular}
\caption{Comparison of OP with SD link with varying $b$ when $d=30$, $N=100$}
\label{diff_b}
\end{table}
\textcolor{blue}{
Table \ref{diff_b} demonstrates the effect of the number of quantization bits $b$ on the OP. Here, one can observe that with a small number of bits itself one can achieve the performance of $b=\infty$ (\textit{i.e.} the no phase error scenario). Furthermore, the improvement in performance with increasing $b$ is not very large for $b>2$. Next, Fig. \ref{d_vary_N_150_s_1}
demonstrates the impact of the position of the $\mathbf{IRS}$ with respect to the positions of $\mathbf{S}$ and $\mathbf{D}$. 
Here one can observe the symmetry of the OP values for different IRS locations about the midpoint of the nodes $\mathbf{S}$ and $\mathbf{D}$. We can also observe that farther the IRS from either of the nodes $\mathbf{S/D}$, larger is the OP. Fig. \ref{vary_N_s_2} demonstrates the number of IRS elements required to match the performance of a system with $N=32$ and $d = 0/90$ when the $IRS$ is located at $d=45$. From Fig. \ref{vary_N_s_2} it is clear that additional $144$ reflector elements are needed when the $\mathbf{IRS}$ is placed mid-way between the nodes $\mathbf{S}$ and $\mathbf{D}$, compared to a system with an IRS located above either of the nodes $\mathbf{S}$ or $\mathbf{D}$. The authors of \cite{tao2020performance,hou2019mimo,wang2020study} also consider the presence of SD link in their analysis but do not take into account the phase error due to $b$ bit phase representation.}
\par \textcolor{blue}{The Kolmogorov-Smirnov (KS) test using the KS statistic is one common test used to decide if a sample comes from a population with a specific distribution \cite{kong2021channel}. The KS statistic is given by $D_{ks}:=\max _{x}\left|F_{}(x)-F_{\text {obs }}(x)\right|$, where $F(x)$ is the proposed approximate CDF and $F_{\text {obs }}(x)$ is the empirical distribution generated using the observations. Let the hypothesis $\vec{H}_0$ denote that the SNR follows the distribution $F(x)$. Then the hypothesis is accepted when $D_{ks}<D_{max}$ and $D_{\max} \approx \sqrt{-(1 / 2 \nu) \ln (\tau / 2)}$, where $\tau$ is the significance level, and $\nu$ is the number of random samples used for generating the empirical distribution \cite{kong2021channel}. We have tabulated the KS statistic for two different values of $N$ in Table \ref{table_ks}. Here, we have taken $\nu=10^6$ and hence for $\tau=0.05$, we have $D_{max}=0.0014$. From the results in Table \ref{table_ks}, we can see that for the method of uni-variate dimension reduction the value of $D_{ks}$ is less than $D_{max}$ for all values of $N$ considered. However, the values of $D_{ks}$ evaluated for the gamma approximations are always larger than the value of $D_{max}$. This is because in both of these approximations, only certain moments of the Gamma distribution are matched with the moments of the exact SNR distribution and hence the approximation of the CDF is not equally tight throughout the support of the SNR. Since the KS test takes into consideration the worst deviation throughout the support, this shows that both the gamma approximations are not as good as the approximation using the uni-variate approximation at all points in the support. This further reinforces the utility of the uni-variate approximation whenever possible. Note that the uni-variate approximation cannot be used in scenarios without an SD link and the gamma approximations will be useful in this regime. The results in the next sub-section shows that the Gamma approximations are fairly close to the empirical CDF for most of the points in the support of interest in the cases without SD link also. }
\subsection{Results without SD link}
\label{without_sd_results}
\textcolor{blue}{
In this sub-section, we study the performance of an IRS-assisted system without an SD link. Table \ref{diff_b_withoutsd} demonstrates the variation in the OP for different values of $b$. Note that for the particular simulation setting considered we need $38$ extra elements to achieve the performance comparable to a system without phase error when only one bit is used to represent the phase. Hence, one can either increase the number of elements or increase $b$ to achieve better performance. Similarly, Fig. \ref{wrtd_23} elucidates the  variation in the OP with respect to $d$. Here also we can observe that the farther the IRS from either of the nodes $\mathbf{S/D}$, the larger is the OP. Furthermore, when the IRS was shifted by $45$ meters \textit{i.e.}, from $d = 0 $ to $d = 45$, three extra bits were required to get similar performance in terms of OP.}
\par \textcolor{blue}{Note that the authors of \cite{atapattu2020reconfigurable,de2021large, kudathanthirige2020performance} also considers scenarios without SD link. In this context, we would like to point out that the work in \cite{atapattu2020reconfigurable} which approximates the square root of SNR as a Gamma RV gives OP expressions which are as tight as ours (in a Rayleigh fading scenario), however, extending their result to scenarios with SD link is not trivial. We had recovered the expressions given by \cite{de2021large} as special cases (presented in corollaries \ref{without_sd} and \ref{without_sd_without_phase_corollary})}.
\begin{table}
\begin{tabular}{|l||*{11}{c|}}\hline
\multirow{5}{*}{}
 $N,\ b$ &\backslashbox{Method}{Threshold($\gamma$)}
 &\makebox[3em]{-25 dB}&\makebox[3em]{-24 dB}&\makebox[3em]{-23 dB}
&\makebox[3em]{-22 dB}&\makebox[3em]{-21 dB}&\makebox[3em]{-20 dB}\\\hline
$N=99$, $b=1$ &Simulated & 0.0005 & 0.0425 & 0.4830 & 0.9612 &  0.9999 & 1\\ \cline{2-8}
& Gamma($k_{mom},\theta_{mom}$)  & 0.0005  & 0.0430 & 0.4830 & 0.9611 & 0.9999  & 1 \\ \cline{2-8}
&Gamma($k_{kl},\theta_{kl}$) & 0.0005  & 0.0428 & 0.4829 & 0.9614 & 0.9999 & 1\\\hline \hline
$N=61,\ b=\infty$ &Simulated &0.0002 &0.0695 &0.7647 & 0.9988  &1 &1 \\  \cline{2-8}
&Gamma($k_{mom},\theta_{mom}$)&0.0002 & 0.0704   &  0.7634 & 0.9990 &1 &1 \\ \cline{2-8}
&Gamma($k_{kl},\theta_{kl}$)  &     0.0002   & 0.0702 &   0.7636  & 0.9990      &1 &1 \\\hline \hline
$N=61,\ b=5$ &Simulated &0.0002 &0.0733 &0.7742 & 0.9989  &1 &1 \\  \cline{2-8}
&Gamma($k_{mom},\theta_{mom}$)&0.0003 &0.0744   & 0.7729 & 0.9991 &1 &1 \\ \cline{2-8}
&Gamma($k_{kl},\theta_{kl}$)  &0.0002   & 0.0742  &  0.7731   & 0.9991      &1 &1 \\\hline \hline
$N=61,\ b=2$ &Simulated & 0.0493 &0.6801 &0.9964 &1 &1 &1 \\  \cline{2-8}
&Gamma($k_{mom},\theta_{mom}$) & 0.0504 &0.6777  & 0.9969 &1 &1 &1 \\ \cline{2-8}
&Gamma($k_{kl},\theta_{kl}$)  &   0.0503   & 0.6778   & 0.9969      &1 &1 &1 \\\hline
\end{tabular}
\caption{Comparison of OP without SD link with varying $b$ when $d=30$}
\label{diff_b_withoutsd}
\end{table}
\begin{table}
\begin{tabular}{|l|*{11}{c|}}\hline
\multirow{12}{*}{}
 $N$ & $d$ & $b$ &\backslashbox{Method}{Threshold($\gamma$)}
&\makebox[3em]{-15 dB}&\makebox[3em]{-10 dB}&\makebox[3em]{-5 dB}&\makebox[3em]{-2 dB}
&\makebox[3em]{0 dB}&\makebox[3em]{2 dB}&\makebox[3em]{5 dB}\\\hline \hline
5 &0 & 1 &Simulated &0.1143 &0.2899 &0.6288 &0.8547  & 0.9527  &0.9925   &1.0000    \\\cline{4-11}
& & & Uni-variate Approx. & 0.1141 & 0.2903  & 0.6288  & 0.8547 & 0.9528 & 0.9925 & 1.0000 \\\cline{1-11}
5 &0 &5 & Simulated &0.1045 &0.2774 &0.6177 &0.8484  & 0.9499 & 0.9919  &0.9999  \\\cline{4-11}
    & & & Uni-variate Approx. & 0.1045 & 0.2778 & 0.6175 & 0.8482 & 0.9500  & 0.9919  & 0.9999 \\\cline{1-11}
5 & 45&1 &Simulated & 0.1293 & 0.3085  &0.6449 &0.8637  & 0.9564 & 0.9933 &1.0000    \\\cline{4-11}
&&& Uni-variate Approx. & 0.1291  & 0.3089  & 0.6449 & 0.8636 &  0.9565 & 0.9932  & 1.0000  \\\cline{1-11}
    5 &45 &5 &Simulated &0.1274 &0.3062 &0.6429  &0.8626  &0.9559 &0.9932   &1.0000    \\\cline{4-11}
    & & & Uni-variate Approx. & 0.1272 & 0.3066 & 0.6429 & 0.8625 & 0.9560 & 0.9931  & 1.0000 \\\cline{1-11} \hline \hline
150 & 0 & 1 &Simulated &0 &0 &0.0383 &0.3062 & 0.5905 & 0.8532   & 0.9951    \\\cline{4-11}
& & & Uni-variate Approx. & 0 &  0 & 0  & 0.1468 & 0.5729 & 0.9312   & 1.0000  \\\cline{1-11} 
150 &0 &5 &Simulated &0 &0 &0 &0.0141 &0.2144 &0.5722  &0.9630    \\\cline{4-11}
    & & & Uni-variate Approx. & 0 & 0 & 0 & 0.0139 & 0.2145 & 0.5727   & 0.9624  \\\cline{1-11}
150 & 45 & 1 &Simulated &0.0458 &0.1949 & 0.5354  &0.7998  &0.9282 &0.9869  &0.9999    \\\cline{4-11}
&&& Uni-variate Approx. & 0.0378 & 0.1909 & 0.5371 & 0.8018  & 0.9303  & 0.9878    & 0.9999   \\\cline{1-11}
150 &45 &5 &Simulated &0.0129 &0.1337 &0.4666  & 0.754  &0.9055 & 0.9814 &0.9998  \\\cline{4-11}
  & & & Uni-variate Approx. &0.0129 &0.1337 & 0.4673  &0.7526 & 0.9052 &0.9811  & 0.9998 \\\cline{1-11}
\end{tabular}
\caption{OP with SD link}
\label{table_ref_forimpact_with_sd}
\end{table}
\begin{table}
\begin{tabular}{|l|*{11}{c|}}\hline
\multirow{12}{*}{}
$N$ &$d$ &$b$ &\backslashbox{Method}{Threshold($\gamma$)}
&\makebox[3em]{-49 dB}&\makebox[3em]{-38 dB}&\makebox[3em]{-36 dB}&\makebox[3em]{-34 dB} &\makebox[3em]{-32 dB}&\makebox[3em]{-30 dB}\\\hline
5 & 0 & 1 &Simulated &0.0365  &0.1224  & 0.3301 &0.6541  &0.9149 & 0.9939    \\\cline{4-10}
&&&Gamma Moment matching &0.0367 & 0.1234  &0.3311 &0.6540  &0.9147 &0.9942      \\\cline{1-10}
    5 &0 &5 &Simulated &0.0000 &0.0006  &0.0117 &0.1193  &0.5104 & 0.9137     \\\cline{4-10}
    &&&Gamma Moment matching &0.0001 & 0.0012 & 0.0155 & 0.1253 &0.5047 &  0.9136    \\\cline{1-10}
    $N$ &$d$ &$b$ &\backslashbox{Method}{Threshold($\gamma$)}
&\makebox[3em]{-55 dB}&\makebox[3em]{-53 dB}&\makebox[3em]{-51 dB}&\makebox[3em]{-49 dB} &\makebox[3em]{-47 dB}&\makebox[3em]{-45 dB}\\\hline
5 & 45 & 1 &Simulated & 0.0322 & 0.1097  &0.3041 & 0.6234   &0.8992 &0.9917      \\\cline{4-10}
&&& Gamma Moment matching &0.0326 &0.1115  & 0.3062  &  0.6236 &  0.8988  & 0.9920    \\\cline{1-10}
    5 &45 &5 &Simulated &0.0000 &0.0004  & 0.0090 & 0.0999 &0.4648 & 0.8918    \\\cline{4-10}
    &&&Gamma Moment matching &0.0000 &0.0009  &0.0125 &  0.1063 &0.4597 &  0.8909   \\\cline{1-10}
$N$ &$d$ &$b$ &\backslashbox{Method}{Threshold($\gamma$)}
&\makebox[3em]{-11 dB}&\makebox[3em]{-10 dB}&\makebox[3em]{-9 dB}&\makebox[3em]{-7 dB} &\makebox[3em]{-6 dB}&\makebox[3em]{-5 dB}\\\hline
100 & 0 & 1 &Simulated &0.0221 &0.3655  &0.9258 &0.9996  &1  & 1   \\\cline{4-10}
&& &Gamma Moment matching &0.0224  &0.3648 &0.9259 &0.9996 &1  &1  \\\cline{1-10}
100 &0 &5 &Simulated & 0        & 0   &  0  & 0.0019  &  0.4460   & 0.9973      \\\cline{4-10}
  && &Gamma Moment matching     &0    & 0     &0  & 0.0021   & 0.4438   & 0.9975    \\\cline{1-10}
  $N$ &$d$ &$b$ &\backslashbox{Method}{Threshold($\gamma$)}
&\makebox[3em]{-25 dB}&\makebox[3em]{-24 dB}&\makebox[3em]{-23 dB}&\makebox[3em]{-22 dB} &\makebox[3em]{-21 dB}&\makebox[3em]{-20 dB}\\\hline
100 & 45 &1 &Simulated     &  0.2554   &  0.8665  & 0.9987 &1 &1 &1 \\\cline{4-10}
&&& Gamma Moment matching & 0.2556   &0.8664   &0.9987 &1 &1 &1 \\\cline{1-10}
100 &45 &5 &Simulated &0  & 0 &0  & 0.0003 &0.2577  & 0.9875     \\\cline{4-10}
&&&Gamma Moment matching &0 &0  &0   & 0.0004 &0.2569   & 0.9878  \\\cline{1-10}
\end{tabular}
\caption{OP without SD link}
\label{table_ref_forimpact_without_sd}
\end{table}
\begin{table}
	\begin{minipage}{0.5\linewidth}
		\centering
		\begin{tabular}{|l||*{5}{c|}}\hline
 Parameter
 &Impact of $N$ &Impact of $b$ &Impact of $d$\\\hline \hline 
$N$ is small  &NA & less &less \\ \hline
$N$ is large &NA & more &more \\\hline \hline
$b$ is small & less      & NA  & less\\ \hline
$b$ is large & more & NA & more \\\hline \hline
$d$ close to $\mathbf{S/D}$ & more & more & NA\\ \hline
$d$ away from $\mathbf{S/D}$ & less &less & NA \\\hline
\end{tabular}
\caption{Summary of impact of $N$, $b$ and $d$ on OP of}
	\label{impact_on_OP}
	\end{minipage}\hfill
	\begin{minipage}{0.45\linewidth}
	\centering
	\includegraphics[width=\textwidth]{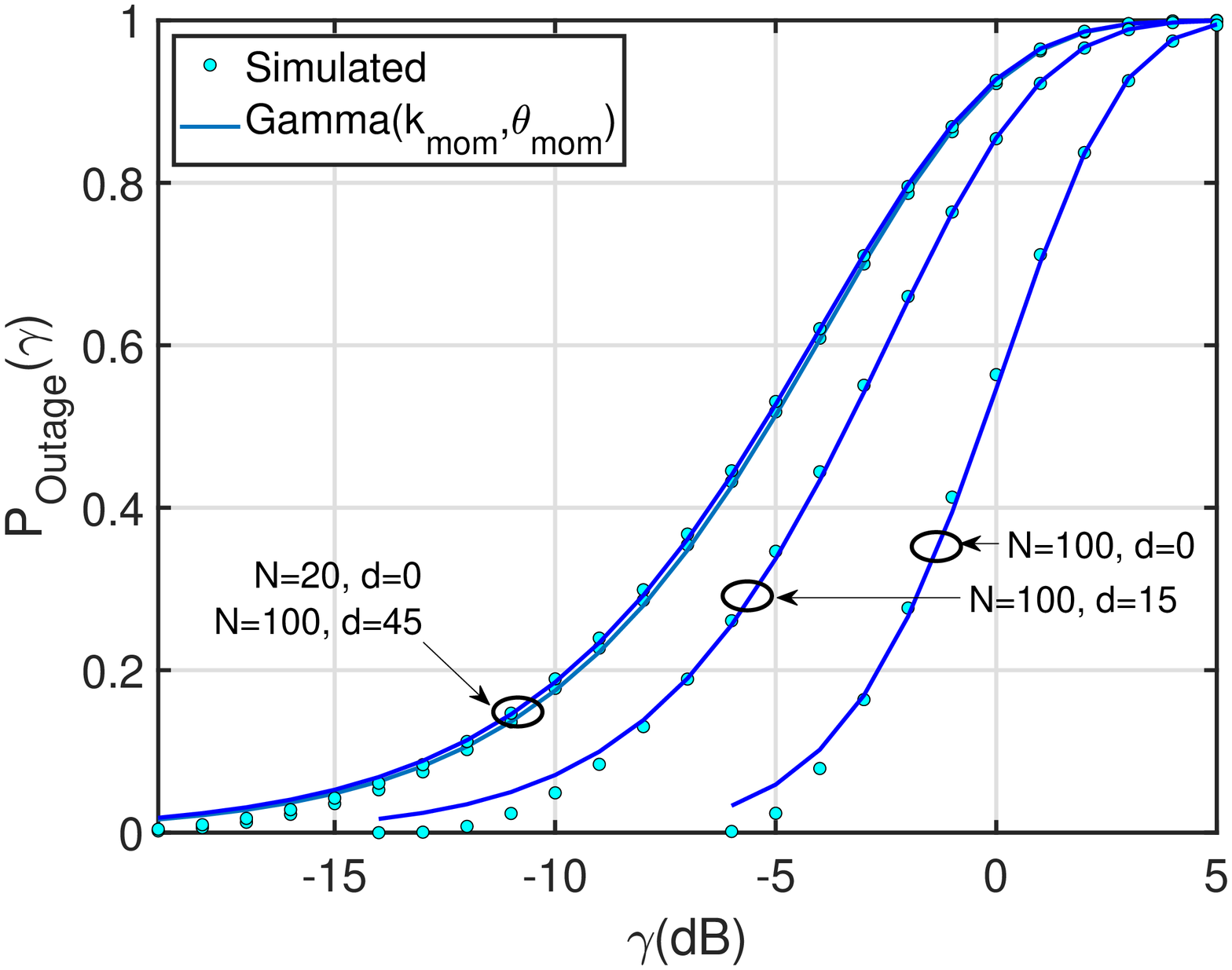}
  \captionof{figure}{Impact of $d$ on the OP without SD link for $\mathcal{S}_{2}$, $N=100$.}
      \label{wrtd_23}
	\end{minipage}
\end{table}
\subsection{Key Inferences}
\textcolor{blue}{In this sub-section, we discuss the key inferences drawn from the results  presented in Section \ref{with_sd_results} and \ref{without_sd_results}. Table \ref{table_ref_forimpact_with_sd} and \ref{table_ref_forimpact_without_sd} compares the values of OP for different values of $N$, $d$, and $b$ for the scenarios with and without SD link, respectively. Since the uni-variate approximation provides the best performance among the three approximations proposed (for scenarios with SD link), we have included only uni-variate approximation in Table \ref{table_ref_forimpact_with_sd}. Similarly, we have only included the moment matching results for the scenarios without SD link. From Table \ref{table_ref_forimpact_with_sd} we can observe that the reduction in OP as $b$ increases from $1$ to $5$ is larger at $N=100$ than at $N=5$. Similar observations can also be made for the cases without SD link from Table \ref{table_ref_forimpact_without_sd}. Thus, we conclude that an IRS with a small number of reflector elements is less sensitive to the number of bits used for representing the phase. Similarly, the values of OP in Tables \ref{table_ref_forimpact_with_sd} and \ref{table_ref_forimpact_without_sd} show that the impact of $d$ on the OP is also larger for large values of $N$. Here, we can also observe that the impact of $b$ and $N$ on the OP increases as the IRS moves closer to either of the nodes $S$ or $D$. Furthermore, we can observe from Tables \ref{table_ref_forimpact_with_sd} and \ref{table_ref_forimpact_without_sd} that the impact of phase errors is larger in a system without an SD link. We can observe that the path loss from the SD link ($d_{SD}^{-\beta}$) is lower than the path loss in the link via the IRS ($d_{SR}^{-\beta}d_{RD}^{-\beta}$)\footnote{\textcolor{blue}{In a triangle, sum of any two sides ($d_{SR} + d_{RD}$) is greater than the third side ($d_{SD}$) and ($d_{SR} \times d_{RD} \ge d_{SR} + d_{RD}$) whenever $d_{SR}>1$ and $d_{RD}>1$. Hence, $d_{SR} \times d_{RD}$ is greater than ($d_{SD}$) when $d_{SR}>1$ and $d_{RD}>1$.}}. Thus, the contribution of the SD link signal to the SNR is larger when compared to the links via the IRS and hence the effect of the SD link can be dominant, especially in the small $N$ regime. We can also observe that the outage performance is more sensitive to the variations in parameters like $b$ and $d$ in the no SD link scenarios when compared to the performance in scenarios with an SD link. This is because whenever the SD link is not in the outage and $d_{SD}$ is fixed, the SD link is the dominant term in the SNR expression, and a small decrease in the terms contributed by the links via the IRS will not degrade the SNR drastically. This, in turn, ensures that the outage does not increase drastically. We have summarised the above observations in Table  \ref{impact_on_OP}.}. \color{blue}
\par Next, in tables \ref{Table:LIS_Res_RelayVsLIS} and \ref{Table:LIS_Res_MisoVsLIS} we compare the OP of the IRS aided system with the performance of a system with a decode and forward (DF) relay and one with multi-antenna source node respectively. Here, we assume that the total transmit power available at the source and relay nodes together and the power available at the multi antenna source node is same as the power available at the source node of the IRS aided system. The OP at the destination of a DF relay ($RL$) assisted system for an threshold ${\gamma}$ is given by \cite[eqn (16)]{laneman2004cooperative}
  \begin{equation}
  P_{outage}^{DF}=\mathbb{P}\left[\min \lbrace |h^{SR}|^2 ,|{h}^{SD}|^2+|{h}^{RD}|^2 \rbrace <\frac{\gamma}{\gamma_s}\right].
 \end{equation} where $h^{SD}$, $h^{SR}$, $h^{RD}$ denote the channel coefficients between the nodes $S$ and $D$, $S$ and $RL$ and $RL$ and $D$ respectively. Here, we assume, ${h}^{SD}$ $\sim$ $\kappa$-$\mu$ $\left(\kappa_{SD}, \mu_{SD}\right)$, ${h}^{SR}$ $\sim$ $\kappa$-$\mu$ $\left(\kappa_{SR}, \mu_{SR}\right)$ and ${h}^{RD}$ $\sim$ $\kappa$-$\mu$ $\left(\kappa_{RD}, \mu_{RD}\right)$. Table \ref{Table:LIS_Res_RelayVsLIS} shows the corresponding values of OP when the relay is present at $(d,h)$ in a system similar to the one shown in Fig \ref{set_up}. From this table, we can see that the performance of the DF relay system depends upon the relay location and with a considerable number of IRS elements, we can always achieve better performance using an IRS aided communication system. Furthermore, the IRS aided system circumvents the need for computational capability and power availability for signal decoding at the IRS unlike the DF relay system. Next, table \ref{Table:LIS_Res_MisoVsLIS} compares the OP of an IRS aided system with an MISO system with $M=4$ antennas at the source node using maximal ratio transmission (MRT) for beamforming. Let $\vec{h}^{SD} \in  \mathbb{C}^{M}$ represent the channel between the nodes $S$ and $D$ and  $\left[\vec{h}^{SD}\right]_{m}$ $\sim$ $\kappa$-$\mu$ $\left(\kappa_{sd}, \mu_{sd}\right)$ $\forall  m \in \lbrace 1,\cdots,M \rbrace$. With MRT at S, the received SNR is given by $||\vec{h}^{SD}||^2$. Hence, the OP of the MISO system at a threshold $\gamma$ can be evaluated as
\begin{equation}
  P_{outage}^{MISO}=\mathbb{P}\left[||\vec{h}^{SD}||^2 < \frac{\gamma}{\gamma_s} \right].
 \end{equation}  Unlike the case of relays, here we need an even larger number of IRS elements to beat the performance of the MRT MISO system. However, note that the computational and power requirements at the multi antenna system will be considerably more than the IRS aided system.  
\begin{table}[h!]
     \centering
        \begin{tabular}{|l||*{8}{c|}}\hline
        \backslashbox{Method}{Threshold($\gamma$)}
        &\makebox[3em]{-15 dB}&\makebox[3em]{-10 dB}&\makebox[3em]{-7 dB}&\makebox[3em]{-5 dB}&\makebox[3em]{-2 dB} &\makebox[3em]{0 dB}&\makebox[3em]{2 dB}&\makebox[3em]{5 dB}\\\hline \hline
        Relay at d=0  &0.0015 &0.0683  &0.3826 &0.7268   & 0.9846  &0.9998  &1  &1 \\\hline \hline
        Relay at d=45  &0 & 0.0018  &0.0083 &0.0261   & 0.1609   &0.4901   &  0.9211  & 1.0000 \\\hline \hline
        IRS with $N$ =25 at $d=0$, $b=5$  &0.0192 &0.1472  & 0.3188   & 0.4831  &0.7634    &0.9106  &0.9826 &0.9998 \\\hline \hline
       IRS with $N$ =50 at $d=0$, $b=5$ &0 & 0.0271  &0.1524  &0.3042  &0.6237  &0.8323  &0.9588 &0.9994 \\\hline \hline
        IRS with $N$ =100 at $d=0$, $b=5$ &0 & 0 &0  & 0.0239  &0.2765  & 0.5639  &0.8372 & 0.9941 \\\hline \hline
        \end{tabular}
     \caption{Comparison of DF Relay and IRS assisted systems}
     \label{Table:LIS_Res_RelayVsLIS}
 \end{table}
 \begin{table}[h!]
     \centering
       \begin{tabular}{|l||*{11}{c|}}\hline
\backslashbox{Method}{Threshold($\gamma$)}
&\makebox[3em]{-5 dB}&\makebox[3em]{-2 dB} &\makebox[3em]{0 dB}&\makebox[3em]{2 dB}&\makebox[3em]{3 dB}&\makebox[3em]{4 dB}&\makebox[3em]{5 dB}\\\hline \hline  
MISO with $M=4$ &0.0323   &0.1772  &0.4260     &0.7544 &0.8821& 0.9583 &0.9902 \\\hline \hline
IRS with $N$ =120, $d=0$, $b=5$ &0.0230   &0.1548  &0.4030  &0.7412 &0.8747    &0.9549    &0.9889 \\\hline \hline
IRS with $N$ =270, $d=15$, $b=5$ & 0.0241   &0.1593   & 0.4100    &0.7470    &0.8784    &0.9566    &0.9894 \\\hline \hline
IRS with $N$ =520, $d=30$, $b=5$ & 0.0291  & 0.1778  &0.4369   &0.7669&0.8903    &0.9617    &0.9909 \\\hline \hline
\end{tabular}
     \caption{Comparison of MISO and IRS assisted systems}
     \label{Table:LIS_Res_MisoVsLIS}
 \end{table}
  \cbend 
  \color{black}
\section{Conclusion and Future Work}
\label{conclusion}
This paper studied the OP of an IRS-assisted communication system in the presence of phase errors due to quantization in a $\kappa-\mu$ fading environment. We proposed three different approximations using 1) uni-variate dimension reduction, 2) moment matching and, 3) KL divergence minimization. Our simulation results showed that the derived expressions are tight and can be reliably used for further analysis. We also observed that the uni-variate dimension reduction method provides an accurate approximation for the OP in terms of an integral expression whereas the moment matching and KL divergence minimization results in simple closed-form expressions for the OP. Furthermore, the proposed approximations are highly useful in evaluating the effects of various system parameters on the OP. In this work, we also studied the impact of the parameters like the number of bits available for quantization, the position of IRS w.r.t. source and destination and, the number of elements present at IRS. Since we have provided a  tight approximation to the CDF, we believe that this can be used to study all other performance metrics which are functions of the CDF of SNR. \cbstart \textcolor{blue}{Note that the system model considered can be more generalised and it would be an interesting future work to study systems where the source and/or the destination is equipped with multiple antennas.} \cbend 
\begin{appendices}
\color{blue} \cbstart 
\section{Proof for lemma \ref{Lemma:LIS_Exact}}
\label{proof:LIS_Exact}
Using (\ref{snr_err}) and (\ref{snr_outage}), the OP for threshold $\gamma$ can be evaluated as
\begin{equation}\label{Eq:LIS_ExactCDF1}
    \begin{aligned}
        P_{outage} & = \mathbb{P}\left( \gamma_{s} \vert(|{h}^{SD}| +  \alpha\sum\limits_{i=1}^N\vert[\vec{h}^{SR}]_{i}\vert \vert [\vec{h}^{RD}]_{i}\vert e^{j\Phi_{i}})\vert^2 \le \gamma \right).
    \end{aligned}
\end{equation}
Let  $|{h}^{SD}|, |[\vec{h}^{SR}]_{i}| $ and $|[\vec{h}^{RD}]_{i}|$ be denoted by $g_{SD}, [\mathbf{g}_{SR}]_i $ and $[\mathbf{g}_{RD}]_i$ respectively. After some algebraic manipulations, we can re-write (\ref{Eq:LIS_ExactCDF1}) as
\begin{equation}\label{Eq:LIS_ExactCDF2}
    \begin{aligned}
        P_{outage} & = \mathbb{P} \left(\gamma_{s}\left(g_{SD}  + \alpha\sum\limits_{i=1}^N [\mathbf{g}_{SR}]_i  [\mathbf{g}_{RD}]_i \cos(\Phi_{i})\right)^2+\gamma_{s}\left(\alpha\sum\limits_{i=1}^N [\mathbf{g}_{SR}]_i  [\mathbf{g}_{RD}]_i \sin(\Phi_{i}) \right)^2 \leq \gamma \right) \\
        & = \mathbb{P} \left(\left(\sqrt{\gamma_{s}}  g_{SD}  +  \mathbf{C}^{T}\mathbf{X}\right)^2 + \left(\mathbf{S}^{T}\mathbf{X}\right)^{2}   \leq \gamma \right)
    \end{aligned}
\end{equation}
where $\mathbf{X} = \left[\sqrt{\gamma_{s}}[\mathbf{g}_{SR}]_{1}  [\mathbf{g}_{RD}]_{1} \dots \sqrt{\gamma_{s}}[\mathbf{g}_{SR}]_{N}  [\mathbf{g}_{RD}]_{N}   \right]$, $\mathbf{C} = \left[ \alpha\cos(\Phi_{1}), \dots ,\alpha\cos(\Phi_{N})  \right]^{T} $ and $\mathbf{S} = \left[\alpha \sin(\Phi_{1}) \dots \alpha \sin(\Phi_{N})  \right]^{T} $. Now, the RV of interest is $Y = \left(\sqrt{\gamma_{s}}  g_{SD}  +  \mathbf{C}^{T}\mathbf{X}\right)^2 + \left(\mathbf{S}^{T}\mathbf{X}\right)^{2}$. Note that $Y$ is the sum of the square of two RVs, one of which is again a sum of $N$ random variables ($\mathbf{S}^{T}\mathbf{X}$, where each $X_{i}$ is a double $\kappa-\mu$ RV and $\Phi_{i}$ is a uniformly distributed RV over the interval $ [-2^{-b}\pi,2^{-b}\pi]$). To the best of our knowledge, characterising the p.d.f. of this sum is not straight forward and is not available in the open literature. Similarly, characterising the distribution of the other term $\textit{i.e.} \sqrt{\gamma_{s}}  g_{SD}  +  \mathbf{C}^{T}\mathbf{X}$ is also difficult. To proceed further, we first derive the conditional CDF of $\gamma_{IRS}$ for a particular value of $\mathbf{\Phi} = \boldsymbol{\phi}$ and $\mathbf{X} = \mathbf{x}$ and is given below
\begin{equation}\label{Eq:LIS_CondCDF1}
    \begin{aligned}
        &P_{outage}\vert\left(\mathbf{\Phi} = \boldsymbol{\phi}, \mathbf{X} = \mathbf{x} \right) = \mathbb{P} \left( \sqrt{\gamma_{s}}  g_{SD} \leq \sqrt{\gamma - \left(\mathbf{s}^{T}\mathbf{x}\right)^{2}}- \mathbf{c}^{T}\mathbf{x} \right)\\
        &= \left(  1- \text{Q}_{\mu_{SD}}\left(\sqrt{2\mu_{SD}\kappa_{SD}},\sqrt{2\mu_{SD}\left(1 + \kappa_{SD}\right)} \frac{\left(\sqrt{\gamma - \left(\mathbf{s}^{T}\mathbf{x}\right)^{2}}- \mathbf{c}^{T}\mathbf{x}\right)}{\sqrt{\gamma_{s}}\hat{t}_{SD}} \right) \right) \operatorname{U}\left(\sqrt{\gamma - \left(\mathbf{s}^{T}\mathbf{x}\right)^{2}}- \mathbf{c}^{T}\mathbf{x}\right)
    \end{aligned}
\end{equation}
Now to evaluate the CDF of $\gamma_{IRS}$, we just need to evaluate the expectation of the R.H.S. of (\ref{Eq:LIS_CondCDF1}) with respect to the RVs $\mathbf{\Phi} = \left[\Phi_{1},\dots,\Phi_{N}\right]^{T}$ and  $\mathbf{X} = \left[X_{1},\dots,X_{N}\right]^{T}$, both of which are multivariate vectors with i.i.d. entries. The p.d.f. of $\Phi_{i}$'s are $f_{\Phi_{i}}(\phi_{i}) = 
 \frac{ 2^{b}}{2\pi}, \ \frac{-\pi}{2^{b}} \le \phi_{i} \le \frac{\pi}{2^{b}}$.
The p.d.f. of double $\kappa-\mu$ RV $X_{i}$ is given as\cite[(9)]{Bhargav2018:ProductKappaMu}  
\begin{equation}\label{Eq:LIS_pdfXMeijer}
    \begin{aligned}
        f_{X_{i}}\left(x_{i}\right) &= \frac{2\sqrt{a_{SR} a_{RD}}}{\sqrt{\gamma_{s}}\rho_{SR}\rho_{RD}} \sum_{m=0}^{\infty} \sum_{n=0}^{\infty}\begin{array}{l}
         \frac{\left(\mu_{SR}\kappa_{SR} \right)^{m}}{m! \Gamma\left( \mu_{SR} + m\right)}\frac{\left(\mu_{RD}\kappa_{RD} \right)^{n}}{n!\Gamma\left( \mu_{RD} + n\right)} 
        G_{0, 2}^{2, 0}
\left[\left. a_{SR}a_{RD} \frac{x_{i}^{2}}{\gamma_{s}}\right|_{\left(\mu_{SR}+m - \frac{1}{2}\right),\left(\mu_{RD}+n-\frac{1}{2}\right)}\right]
        \end{array} 
    \end{aligned}
\end{equation}
where $ \rho_{SR} = e^{\mu_{SR}\kappa_{SR}}, \rho_{RD} = e^{\mu_{RD}\kappa_{RD}} $ and $a_{SR} = \frac{\mu_{SR}\left(1+\kappa_{SR}\right) }{\hat{t}^{2}_{SR}}, a_{RD} = \frac{\mu_{RD}\left(1+\kappa_{RD}\right) }{\hat{t}^{2}_{RD}} $. 
Using the identity $ G_{0,2}^{2,0}\left[\left. z\right|_{b,c} \right] = 2 z^{\frac{b+c}{2}}K_{b-c}\left(2\sqrt{z} \right)$,we have
\begin{equation}\label{Eq:LIS_pdfXBessel}
    \begin{aligned}
        f_{X_{i}}\left(x_{i}\right) &= \frac{4\sqrt{a_{SR} a_{RD}}}{\sqrt{\gamma_{s}}\rho_{SR}\rho_{RD}} \sum_{m=0}^{\infty} \sum_{n=0}^{\infty}\begin{array}{l}
         \frac{\left(\mu_{SR}\kappa_{SR} \right)^{m}}{m! \Gamma\left( \mu_{SR} + m\right)}\frac{\left(\mu_{RD}\kappa_{RD} \right)^{n}}{n!\Gamma\left( \mu_{RD} + n\right)} 
         \left(\sqrt{\frac{a_{SR}a_{RD}}{\gamma_{s}}} x_{i}\right)^{\mu_{SR}+\mu_{RD}+m+n-1} \\ \times K_{\mu_{SR}-\mu_{RD} + m-n}\left(2\sqrt{\frac{a_{SR}a_{RD}}{\gamma_{s}}} x_{i} \right)
        \end{array} 
    \end{aligned}
\end{equation}
Thus, the CDF of $\gamma_{IRS}$ is given by
\begin{equation}\label{Eq:LIS_ExactCDF3}
    \begin{aligned}
        P_{outage} &= \int\dots\int\left(  1- \text{Q}_{\mu_{SD}}\left(\sqrt{2\mu_{SD}\kappa_{SD}},\sqrt{2\mu_{SD}\left(1 + \kappa_{SD}\right)} \frac{\left(\sqrt{\gamma - \left(\mathbf{s}^{T}\mathbf{x}\right)^{2}}- \mathbf{c}^{T}\mathbf{x}\right)}{\sqrt{\gamma_{s}}\hat{t}_{SD}} \right) \right)\\ 
        &\hspace{25mm} \times\operatorname{U}\left(\sqrt{\gamma - \left(\mathbf{s}^{T}\mathbf{x}\right)^{2}}- \mathbf{c}^{T}\mathbf{x}\right) \prod_{i=1}^{N}  f_{X_{i}}\left(x_{i}\right) f_{\Phi_{i}}\left(\phi_{i}\right) d x_{1} d \phi_{1} \dots d x_{N} d \phi_{N}
    \end{aligned}
\end{equation}
The result in (\ref{Eq:LIS_ExactCDF_Final}) follows by substituting the p.d.f. expressions of $\Phi_{i}$ and $X_{i}$ in (\ref{Eq:LIS_ExactCDF3}), and this completes the proof.
\section{Proof for Theorem \ref{Thm:LIS_DimReduc}}
\label{proof:LIS_DimReduc}
Consider a random vector $\mathbf{Y} = \left[\mathbf{X} \ \mathbf{\Phi}  \right]^{T}$ then (\ref{Eq:LIS_ExactCDF3}) can be interpreted as
\begin{equation}\label{Eq:LIS_ExactCDF4}
    \begin{aligned}
        P_{outage} &= \mathbb{E}\left[\left(  1- \text{Q}_{\mu_{SD}}\left(\sqrt{2\mu_{SD}\kappa_{SD}},\sqrt{2\mu_{SD}\left(1 + \kappa_{SD}\right)} \frac{\left(\sqrt{\gamma - \left(\mathbf{s}^{T}\mathbf{x}\right)^{2}}- \mathbf{c}^{T}\mathbf{x}\right)}{\sqrt{\gamma_{s}}\hat{t}_{SD}} \right) \right)\right. \\
        &\hspace{20mm}\left.\operatorname{U}\left(\sqrt{\gamma - \left(\mathbf{s}^{T}\mathbf{x}\right)^{2}}- \mathbf{c}^{T}\mathbf{x}\right) \right], \\
        &= \int_{\mathbb{R}^{2N}} g\left( \mathbf{y} \right) f_{\mathbf{Y}}\left( \mathbf{y}\right) d \mathbf{y}.
    \end{aligned}
\end{equation}
Now, using \cite[eq. 20]{Rahman2004:Integral_DimensionReduction} we approximate (\ref{Eq:LIS_ExactCDF4}) as follows
\begin{equation}\label{Eq:LIS_DimReducMain}
    \begin{aligned}
         P_{outage} &\approx \sum_{i = 1}^{2N}\mathbb{E}\left[g\left(\mu_{1}\dots\mu_{i-1},y_{i},\mu_{i+1},\dots,\mu_{2N}\right)\right] - \left( 2N - 1\right)g\left(\mu_{1},\dots,\mu_{2N} \right)
    \end{aligned}
\end{equation}
where 
\begin{equation}\label{Eq:LIS_mean}
    \begin{aligned}
        \mu_{i} = \mathbb{E}\left[ Y_{i}\right] = \begin{cases}\begin{array}{l}
             \frac{\left( \mu_{SR} \right)_{\frac{1}{2}} \left( \mu_{RD} \right)_{\frac{1}{2}}}{\sqrt{a_{SR} a_{RD}}} {}_{1}F_{1}\left(-\frac{1}{2} ; \mu_{SR} ; -\kappa_{SR} \mu_{SR}\right)  \\
             \times {}_{1}F_{1}\left(-\frac{1}{2} ; \mu_{RD} ; -\kappa_{RD} \mu_{RD}\right) 
        \end{array}   \triangleq \mu  &\quad 1\le i \le N
        \\ 0 &\quad N+1 \le i \le 2N.
        \end{cases}
    \end{aligned}
\end{equation}
Now, $g\left(\mu_{1}\dots\mu_{i-1},y_{i},\mu_{i+1},\dots,\mu_{2N}\right)$ can be calculated using (\ref{Eq:LIS_mean}) as follows
\begin{equation}\label{Eq:LIS_g}
    \begin{aligned}
        g\left(\mu_{1}\dots\mu_{i-1},y_{i},\mu_{i+1},\dots,\mu_{2N}\right)
        \\ &\hspace{-5cm}= \begin{cases}
        \begin{array}{l}
            \left[ \left(  1- \text{Q}_{\mu_{SD}}\left(\sqrt{2\mu_{SD}\kappa_{SD}},\sqrt{2\mu_{SD}\left(1 + \kappa_{SD}\right)} \frac{\left(\sqrt{\gamma} - x_{i} - (N-1)\alpha\mu\right)}{\sqrt{\gamma_{s}}\hat{t}_{SD}} \right) \right) \right. \\
             \left. \operatorname{U}\left(\sqrt{\gamma}- x_{i} - (N-1)\alpha\mu \right) \right] \quad 1 \le i \le N
        \end{array}
             \vspace{0.5cm}
            \\
            \begin{array}{l}
                 \left[\left(  1- \text{Q}_{\mu_{SD}}\left(\sqrt{2\mu_{SD}\kappa_{SD}},\sqrt{2\mu_{SD}\left(1 + \kappa_{SD}\right)} \frac{\left(\sqrt{\gamma - \left(\alpha\sin{\left( \phi_{i} \right)}\mu \right)^{2}}- 
        \alpha\cos{\left(\phi_{i}\right)} \mu - (N-1)\alpha\mu\right)}{\sqrt{\gamma_{s}}\hat{t}_{SD}} \right) \right) \right. \\
                 \left. \operatorname{U}\Bigg( \Bigg.\sqrt{\gamma - \left(\alpha\sin{\left( \phi_{i} \right)}\mu \right)^{2}}- 
        \alpha\cos{\left(\phi_{i}\right)} \mu - (N-1)\alpha\mu \Bigg.\Bigg) \right]  \quad N+1 \le i \le 2N,
            \end{array} 
        \end{cases}
    \end{aligned}
\end{equation}
and 
\begin{equation}\label{Eq:LIS_gconst}
    \begin{aligned}
        g\left(\mu_{1},\dots,\mu_{2N}\right) = \left(  1- \text{Q}_{\mu_{SD}}\left(\sqrt{2\mu_{SD}\kappa_{SD}},\sqrt{2\mu_{SD}\left(1 + \kappa_{SD}\right)} \frac{\left(\sqrt{\gamma}- N\alpha\mu\right)}{\sqrt{\gamma_{s}}\hat{t}_{SD}} \right) \right) \operatorname{U}\left(\sqrt{\gamma}- N\alpha\mu \right).
    \end{aligned}
\end{equation}
Finally, the approximation in (\ref{Eq:LIS_DimReducApprox}) is deduced after substituting values from (\ref{Eq:LIS_g}) , (\ref{Eq:LIS_gconst}) into (\ref{Eq:LIS_DimReducMain}), and the
proof is complete.
\section{proof for Theorem \ref{gamma_approx}}
\label{gamma}
Here, we derive the first and second moments of the RV $\gamma_{IRS}$. Note that $\gamma_{IRS}$ can be expanded as follows: 
 \begin{align}
    \frac{\gamma_{IRS}}{\gamma_{s}} = & \underbrace{|{h}^{SD}|^2 + \alpha^2 \sum_{i=1}^N|[\vec{h}^{SR}]_{i}|^2|[\vec{h}^{RD}]_{i}|^2}_{A}+\underbrace{ 
    2\alpha|{h}^{SD}| \sum_{i=1}^N |[\vec{h}^{SR}]_{i}||[\vec{h}^{RD}]_{i}|\cos(\Phi_{i})}_{B}+ \nonumber \\ & 
     \underbrace{2\alpha^2 \sum_{i=1}^{N-1}\sum_{k=i+1}^N|[\vec{h}^{SR}]_{i}||[\vec{h}^{RD}]_{i}||[\vec{h}^{SR}]_{k}||[\vec{h}^{RD}]_{k}|\cos(\Phi_{i}-\Phi_{k})}_{C}.
\end{align}
Hence, the first two moments of $\gamma_{IRS}$ can be evaluated as
$\mathbb{E}\left[\gamma_{IRS}\right] =\gamma_{s} \left( \mathbb{E}[A]+\mathbb{E}[B]+\mathbb{E}[C] \right)$ and $\mathbb{E}\left[\gamma_{IRS}^2\right] =\gamma^{2}_{s}\left( \mathbb{E}\left[A^2\right] + \mathbb{E}\left[B^2\right] + \mathbb{E}\left[C^2\right] + 2 \mathbb{E}\left[AB\right] + 2 \mathbb{E}\left[BC\right] + 2 \mathbb{E}\left[AC\right] \right)$. Next, in order to evaluate the above moments, we substitute $A$, $B$, $C$ and derive individual expectations using the fact that the RVs $\{h^{SD},\left[\vec{h}^{SR}\right]_i,\left[\vec{h}^{RD}\right]_i; i=1,\cdots,N\}$ are independent. In the subsequent paragraphs, we assume
 $s:=\frac{2^b}{\pi}{\sin\left(\frac{\pi}{2^b}\right)}$, $p:=\frac{2^b}{2\pi}{\sin\left(\frac{2\pi}{2^b}\right)}$, $m_{1}^{AB}:=\mathbb{E}\left[|[\vec{h}^{AB}]_i|\right]$, $m_{2}^{AB}:=\mathbb{E}\left[|[\vec{h}^{AB}]_i|^2\right]$, $m_{3}^{AB}:=\mathbb{E}\left[|[\vec{h}^{AB}]_i|^3\right]$ and $m_{4}^{AB}:=\mathbb{E}\left[|[\vec{h}^{AB}]_i|^4\right]$ for  $A,B$ $\in$ \{$\mathbf{S}$,\ $\mathbf{R}$,\ $\mathbf{D}$\}. Let us first evaluate the first moments of RV's $A$, $B$, and $C$. Here, we have,
\begin{equation}
  \begin{aligned}
    \mathbb{E}\left[A\right] &=\mathbb{E}\left[|{h}^{SD}|^2 + \alpha^2 \sum_{i=1}^N|[\vec{h}^{SR}]_{i}|^2|[\vec{h}^{RD}]_{i}|^2\right]
    = m_2^{SD}+N \alpha^2 m_2^{SR} m_2^{RD}.
    \label{meanA}
    \end{aligned}
\end{equation}
Similarly, we have,
\begin{equation}
\begin{aligned}
    \mathbb{E}\left[B\right] &=\mathbb{E}\left[2\alpha|{h}^{SD}| \sum_{i=1}^N |[\vec{h}^{SR}]_{i}||[\vec{h}^{RD}]_{i}|\cos\left(\Phi_{i}\right)\right] =2N\ \alpha s m_1^{SD} m_1^{SR} m_1^{RD},
    \end{aligned}
    \label{meanB}
\end{equation} 
where  $ s = \mathbb{E}\left[\cos(\Phi_{i})\right]$. 
\begin{equation}
    \begin{aligned}
    \mathbb{E}[C] &=\mathbb{E}\left[2\alpha^2 \sum_{i=1}^{N-1}\sum_{k=i+1}^N|[\vec{h}^{SR}]_{i}||[\vec{h}^{RD}]_{i}||[\vec{h}^{SR}]_{k}||[\vec{h}^{RD}]_{k}|\cos\left(\Phi_{i}-\Phi_{k}\right)\right]\\
    &=N(N-1) \alpha^2 s^2 (m_1^{SR})^2 (m_1^{RD})^2,
    \label{meanC}
    \end{aligned}
\end{equation}
where $\mathbb{E}\left[\cos\left(\Phi_{i}-\Phi_{k}\right)\right]=\left(\frac{2^b}{{\pi}}\sin\left(\frac{\pi}{2^b}\right)\right)^2=s^2$.
Similarly, the second moments of the RV $A$, $B$, and $C$ are derived as follows,
\begin{equation}
\begin{aligned}
    \hspace{-1cm}\mathbb{E}\left[A^2\right]&=\mathbb{E}\left[\left(|{h}^{SD}|^2 + \alpha^2 \sum_{i=1}^N|[\vec{h}^{SR}]_{i}|^2|[\vec{h}^{RD}]_{i}|^2\right)^2\right]\\
    &=m_4^{SD}+2\alpha^2N m_2^{SR}m_2^{RD}m_2^{SD}+N \alpha^4 m_4^{SR}m_4^{RD}+N(N-1) \alpha^4 (m_2^{SR})^2 (m_2^{RD})^2,
    \end{aligned}
\end{equation}

\begin{equation}
\begin{aligned}
     \hspace{-1cm}\mathbb{E}[B^2]=&\mathbb{E}\left[\left(2\alpha|{h}^{SD}| \sum_{i=1}^N |[\vec{h}^{SR}]_{i}||[\vec{h}^{RD}]_{i}|\cos(\Phi_{i})\right)^2\right]\\
    & = 4\alpha^2 m_2^{SD}\left[Nm_2^{SR}m_2^{RD}\frac{1+p}{2} +N(N-1)s^2(m_1^{SR})^2 (m_1^{RD})^2 \right],
\end{aligned}
\end{equation}
where $\mathbb{E}\left[\cos^2(\Phi_{i})\right]=\frac{1}{2}+\frac{2^b}{2\pi}\sin(\frac{\pi}{2^b})\cos(\frac{\pi}{2^b})=\frac{1+p}{2}$. Next, we have
%\begin{equation}
\begin{align}
    \mathbb{E}[C^2]&=\mathbb{E}\left[\left(2\alpha^2 \sum_{i=1}^{N-1}\sum_{k=i+1}^N|[\vec{h}^{SR}]_{i}||[\vec{h}^{RD}]_{i}||[\vec{h}^{SR}]_{k}||[\vec{h}^{RD}]_{k}|\cos(\Phi_{i}-\Phi_{k})\right)^2\right] \nonumber \\
     &=4\alpha^4\sum_{i=1}^{N-1}\sum_{k=i+1}^N\left(\mathbb{E}\left[|[\vec{h}^{SR}]_{i}|^2\right]\right)^2\left(\mathbb{E}\left[|[\vec{h}^{RD}]_{i}|^2\right]\right)^2\mathbb{E}\left[\cos^2(\Phi_{i}-\Phi_{k}))\right]  \nonumber \\
     &+4\alpha^4\sum_{j=1}^N\sum_{\substack{i=1 \\  i \ne j \ne k \ne l}}^{N-1}\sum_{k=i+1}^N\sum_{l=j+1}^N \left(\mathbb{E}\left[|\left[\vec{h}^{SR}\right]_{i}|\right]\right)^4\left(\mathbb{E}\left[|\left[\vec{h}^{RD}\right]_{i}|\right]\right)^4\mathbb{E}^2\left[\cos\left(\Phi_{i}-\Phi_{k}\right)\right]  \label{equation_c2}\\
     &+ 8\alpha^4 \sum_{j=1}^N\sum_{\substack{i=1 \\  j=i \ne k \ne l }}^{N-1}\sum_{k=i+1}^N\sum_{l=j+1}^N \Big[ \mathbb{E}\left[|[\vec{h}^{SR}]_{k}|\right]\mathbb{E}\left[|[\vec{h}^{RD}]_{k}|\right]\mathbb{E}\left[|[\vec{h}^{SR}]_{l}|\right]\mathbb{E}\left[|[\vec{h}^{RD}]_{l}|\right] \Big. \nonumber \\
     &\hspace{4cm} \Big. \mathbb{E}\left[|[\vec{h}^{SR}]_{i}|^2\right]\mathbb{E}\left[|[\vec{h}^{RD}]_{i}|^2 \right]\mathbb{E}\left[\cos(\Phi_{i}-\Phi_{k})\cos(\Phi_{i}-\Phi_{l})\right] \Big].\nonumber
\end{align}
%\end{equation}
Using $ \mathbb{E} \left[\cos^2(\Phi_{i}-\Phi_{k})\right]=\frac{1}{2}+\frac{1}{2}\left(\frac{2^b}{2\frac{\pi}{2^b}}\sin\left(2\frac{\pi}{2^b}\right)\right)^2=\frac{1+p^2}{2}$ and  $\mathbb{E}\left[\cos\left(\Phi_{i}-\Phi_{k}\right)\cos(\Phi_{i}-\Phi_{l})\right]=\frac{1}{2}s^2+\frac{1}{2}p s^2$,
(\ref{equation_c2}) can be re-written as,
\begin{equation}
    \begin{aligned}
     \mathbb{E}\left[C^2\right]&=N(N-1)\alpha^4\left[2(N-2)m_2^{SR}m_2^{RD}s^2\left(1+p\right)(m_1^{SR})^2(m_1^{RD})^2+(m_2^{SR})^2(m_2^{RD})^2(1+p^2)\right.\\
     &\left.+s^4(N-2)(N-3)(m_1^{SR})^4(m_1^{RD})^4\right].
    \end{aligned}
\end{equation}
Next, we derive the expectation of RV's $AB$, $BC$, and $AC$, 
\begin{align}
 \mathbb{E}[AB] = \mathbb{E}\left[\left(|{h}^{SD}|^2 +  \alpha^2\sum_{i=1}^N|[\vec{h}^{SR}]_{i}|^2|[\vec{h}^{RD}]_{i}|^2 \right) \times \left( 2\alpha|{h}^{SD}| \sum_{i=1}^N |[\vec{h}^{SR}]_{i}||[\vec{h}^{RD}]_{i}|\cos(\Phi_{i})\right)\right].
\end{align}
We can rewrite the above equation as (\ref{expectation_ab}).
\begin{align}
 \mathbb{E}\left[AB\right] = & 2\alpha \mathbb{E}\left[|{h}^{SD}|^3\right] \sum_{i=1}^N \mathbb{E}\left[|[\vec{h}^{SR}]_{i}|\right]\mathbb{E}\left[|[\vec{h}^{RD}]_{i}|\right] \mathbb{E}\left[\cos(\Phi_{i})\right] \nonumber \\& 
 +4\alpha^3 \mathbb{E}\left[|{h}^{SD}|\right] \sum_{j=1}^{N-1}\sum_{i=j+1}^N \mathbb{E}\left[|[\vec{h}^{SR}]_{i}|^2\right]\mathbb{E}\left[|[\vec{h}^{RD}]_{i}|^2 \right]  \mathbb{E}\left[|[\vec{h}^{SR}]_{i}|\right]\mathbb{E}\left[|[\vec{h}^{RD}]_{i}|\right]\mathbb{E}\left[\cos(\Phi_{i})\right] \nonumber \\&
 +2\alpha^3 \mathbb{E}\left[|{h}^{SD}|\right] \sum_{i=1}^N \mathbb{E}\left[|[\vec{h}^{SR}]_{i}|^3\right]\mathbb{E}\left[|[\vec{h}^{RD}]_{i}|^3\right] \mathbb{E}\left[ \cos(\Phi_{i})\right]).
 \label{expectation_ab}
\end{align}
The above expression can be evaluated as,
 \begin{align}
 \mathbb{E}[AB] = &2N\alpha m_3^{SD} m_1^{SR} m_1^{RD}s +2 \alpha^3 N m_1^{SD} m_3^{SR} m_3^{RD} s+4 \alpha^3 m_1^{SD} m_1^{SR} m_1^{RD} m_2^{SR} m_2^{RD}\frac{N(N-1)}{2}s.
 \end{align}
Next, we consider the expectation of the term $BC$,
\begin{align}
 \mathbb{E}[BC] = & \mathbb{E} \left[ 2\alpha|{h}_{SD}| \sum_{i=1}^N |[\vec{h}_{SR}]_{i}||[\vec{h}_{RD}]_{i}|\cos(\Phi_{i}) \times \nonumber \right. \\ 
 & \left. 2\alpha^2\sum_{i=1}^{N-1}\sum_{k=i+1}^N|[\vec{h}_{SR}]_{i}||[\vec{h}_{RD}]_{i}||[\vec{h}_{SR}]_{k}||[\vec{h}_{RD}]_{k}|(\cos(\Phi_{i}-\Phi_{k})) \right].
\end{align}
We can write the above equation as 
\begin{align}
 \mathbb{E}[BC]= 2\alpha^3 m_1^{SD} \left[m_1^{SR} \right]^3\left[m_1^{RD}  \right]^3
N(N-1)(N-2)s^3 +2\alpha^3s m_1^{SD}  m_2^{SR}m_2^{RD} m_1^{SR} m_1^{RD} N(N-1)(1+p).
\end{align}
Finally, we evaluate the expectation of the last term $AC$,
\begin{align}
 \mathbb{E}[AC] = & \mathbb{E}\left[\left(|{h}^{SD}|^2 +  \alpha^2\sum_{i=1}^N|[\vec{h}^{SR}]_{i}|^2|[\vec{h}^{RD}]_{i}|^2 \right) \right.\times \nonumber \\ & \left. \left(2\alpha^2\sum_{i=1}^{N-1}\sum_{k=i+1}^N|[\vec{h}^{SR}]_{i}||[\vec{h}^{RD}]_{i}||[\vec{h}^{SR}]_{k}||[\vec{h}^{RD}]_{k}|\cos(\Phi_{i}-\Phi_{k}) \right)\right].
\end{align}
Again, we can rewrite the above equation as
\begin{align}
 \mathbb{E}[AC] = & \mathbb{E}[|{h}^{SD}|^2] \mathbb{E}[C]+ 4\alpha^4\sum_{i=1}^{N-1} \sum_{k=i+1}^N \mathbb{E}[|[\vec{h}^{SR}]_{j}|^3]\mathbb{E}[|[\vec{h}^{RD}]_{j}|^3]\mathbb{E}[|[\vec{h}^{SR}]_{j}|]\mathbb{E}[|[\vec{h}^{RD}]_{j}|] \mathbb{E}[\cos(\Phi_{i}-\Phi_{k})] \nonumber \\&
  +2\alpha^4\sum_{j=1}^N\sum_{\substack{i=1 \\ i \ne j \ne k}}^{N-1}\sum_{k=i+1}^N  \mathbb{E}[|[\vec{h}^{SR}]_{j}|^2]\mathbb{E}[|[\vec{h}^{RD}]_{j}|^2]\mathbb{E}[|[\vec{h}^{SR}]_{j}|]^2\mathbb{E}[|[\vec{h}^{RD}]_{j}|]^2\mathbb{E}[\cos(\Phi_{i}-\Phi_{k})].
\end{align}
Substituting for the expectations, we get
\begin{align}
 \mathbb{E}[AC] = &m_2^{SD} \left(2 \alpha^2(m_1^{SR})^2 (m_1^{RD})^2 \frac{N(N-1)}{2}s^2\right)+4 \alpha^4 m_3^{SR} m_3^{RD} m_1^{SR} m_1^{RD} \frac{N(N-1)}{2}s^2\\
 &+2 \alpha^4 (m_1^{SR})^2  (m_1^{RD})^2 m_2^{SR} m_2^{RD} \frac{N(N-1)(N-2)}{2}s^2.
\end{align}
Thus, the first and second moments are given by
\begin{equation}
    \mathbb{E}\left[{\gamma_{IRS}}\right] =\gamma_{s}\left( m_2^{SD}+N \alpha^2 m_2^{SR} m_2^{RD}+2Ns\ \alpha m_1^{SD} m_1^{SR} m_1^{RD}+N(N-1) \alpha^2(m_1^{SR})^2 (m_1^{RD})^2 s^2\right).
    \label{mean_1_general_fad}
\end{equation}
\begin{align} \label{mean_2_general_fad}
   &\mathbb{E}\left[{\gamma_{IRS}}^2\right] =\gamma^2_{s}\left\lbrace m_4^{SD}+2\alpha^2N m_2^{SR}m_2^{RD}m_2^{SD}+N \alpha^4 m_4^{SR}m_4^{RD}+N(N-1) \alpha^4 (m_2^{SR})^2 (m_2^{RD})^2 \right.\\ \nonumber
 & \left.+ 4\alpha^2 m_2^{SD}\left[N m_2^{SR}m_2^{RD}\frac{1+p}{2} +N(N-1)s^2(m_2^{SR})^2 (m_2^{RD})^2 \right] + N(N-1)\alpha^4 \left[ 2(N-2) s^2(1+p) \right.  \right.\\ \nonumber
 &\left. \left. m_2^{SR}m_2^{RD}(m_1^{SR})^2(m_1^{RD})^2 +(m_2^{SR})^2(m_2^{RD})^2(1+p^2)+s^4(N-2)(N-3)(m_1^{SR})^4(m_1^{RD})^4\right] \right.\\ \nonumber
 &\left.+ 4N\alpha s \left[ m_3^{SD} m_1^{SR} m_1^{RD} + \alpha^2  m_1^{SD} m_3^{SR} m_3^{RD} +(N-1) \alpha^2 m_1^{SD} m_1^{SR} m_1^{RD} m_2^{SR} m_2^{RD}
 \right] \right. \\ \nonumber
 &\left.+ 4N(N-1)\alpha^3s \left[ m_1^{SD} \left(m_1^{SR} \right)^3\left(m_1^{RD}  \right)^3
(N-2)s^2  +(1+p) m_1^{SD}  m_2^{SR}m_2^{RD} m_1^{SR} m_1^{RD} 
\right]\right. \\ \nonumber
 &\left.+ 2 \alpha^2 N(N-1) s^2 m_1^{SR} m_1^{RD}  \left[m_2^{SD} m_1^{SR} m_1^{RD} +\alpha^2 \left(2 m_3^{SR} m_3^{RD}  + m_1^{SR}  m_1^{RD} m_2^{SR} m_2^{RD} (N-2)\right)
 \right]\right\rbrace.
\end{align}
Note that the values of $m_{1}^{AB},m_{2}^{AB},m_{3}^{AB}$ and $m_{4}^{AB}$ depend upon the channel fading statistics and for the case of $\kappa-\mu$ fading these moments can be evaluated using \cite[(3)]{Bhargav2018:ProductKappaMu}

% distribution with parameters $\kappa_{AB}$ and $\mu_{AB}$ we have, 
% \begin{equation}
% \mathbb{E}\left[X^{p}\right]=\frac{\Omega_{AB}^{\frac{p}{2}} \Gamma\left(\mu_{AB}+\frac{p}{2}\right) \mathrm{e}^{-\kappa_{AB} \mu_{AB}}}{\Gamma\left(\mu_{AB}\right)\left[\left(1+\kappa_{AB}\right) \mu_{AB}\right]^{\frac{p}{2}}}\  _1F_{1}\left(\mu_{AB}+\frac{p}{2} ; \mu_{AB} ; \kappa_{AB} \mu_{AB}\right),
% \label{kappa_mu_k_th_moment}
% \end{equation}where $\Omega_{AB}:=\mathbb{E}[X^2]$ and is equal to $d_{AB}^{-\beta}$ in our case. Given that $|[\vec{h}^{AB}]_i|$ follows the $\kappa-\mu$ distribution with parameters $\kappa_{AB}$ and $\mu_{AB}$, we can substitute the values of $m_{1}^{AB}:=\mathbb{E}\left[|[\vec{h}^{AB}]_i|\right]$, $m_{2}^{AB}:=\mathbb{E}\left[|[\vec{h}^{AB}]_i|^2\right]$, $m_{3}^{AB}:=\mathbb{E}\left[|[\vec{h}^{AB}]_i|^3\right]$ and $m_{4}^{AB}:=\mathbb{E}\left[|[\vec{h}^{AB}]_i|^4\right]$ in (\ref{mean_1_general_fad}) and (\ref{mean_2_general_fad}) to evaluate the first and second moment of SNR in different fading scenarios.  
\cbend 
\color{black}
\section{Proof for Theorem \ref{thm_kl_div_min}} \label{proof_kl_min}
Here, we proceed with steps similar to the KL divergence minimization used by the authors of  \cite{srinivasan2017approximate}. Let $p(\gamma)$ and $q(\gamma)$ respectively represent the pdf of the SNR and the Gamma distribution
that minimizes the KL divergence between $p(\gamma)$ and all the Gamma distributions i.e.,
\begin{align}
q(\gamma)=\underset{q(\gamma)}{\operatorname{argmin}}  \ \text{KL} (p(\gamma) \| q(\gamma)) &=\underset{q(\gamma)}{\operatorname{argmax}} \int p(\gamma)[\ln (q(\gamma))-\ln (p(\gamma))] d \gamma, \nonumber \\
&=\underset{q(\gamma)}{\operatorname{argmax}} \int p(\gamma) \ln (q(\gamma)) d \gamma.
\end{align}
Here we have, $q(\gamma) = \frac{\theta_{kl}^{-k_{kl}}}{\Gamma[k_{kl}]} \gamma^{k_{kl}-1} \exp\left( \frac{-\gamma}{\theta_{kl}}\right)$, where $k_{kl}$, $\theta_{kl}$ are respectively the shape and scale parameter of the Gamma distribution. Thus,
\begin{align}
q(\gamma)  & = \underset{q(\gamma)}{\operatorname{argmax}} \int p(\gamma) \left(-k_{kl} \log(\theta_{kl}) - \log(\Gamma[k_{kl}]) + (k_{kl}-1)\log(\gamma) - \frac{\gamma}{\theta_{kl}} \right) \ d\gamma \\
& = \underset{q(\gamma)}{\operatorname{argmax}} -k_{kl} \log(\theta_{kl}) - \log(\Gamma[k_{kl}]) + (k_{kl}-1) \mathbb{E}[\log(\gamma)] - \frac{\mathbb{E}[\gamma]}{\theta_{kl}}.
\label{gamma_kl_opti}
\end{align}
Now, the parameters $k_{kl}$ and $\theta_{kl}$ can be identified by differentiating (\ref{gamma_kl_opti}) with respect to $k_{kl}$ and $\theta_{kl}$ and then equating each of the expression to zero. Thus, we have
\begin{align}
\mathbb{E}[\log(\gamma_{IRS})] & = \log(\theta_{kl}) + \psi(k_{kl}), \label{log_e_1} \\
\mathbb{E}[\gamma_{IRS}] &= k_{kl} \times \theta_{kl}.  \label{e_1}
\end{align}
One can observe that (\ref{log_e_1}) and (\ref{e_1}) are equivalent to matching the first moment of $\gamma_{IRS}$ and the first moment of $\log(\gamma_{IRS})$ to the corresponding moments of a gamma RV. Thus, the probability of outage for a threshold $\gamma$ is obtained by evaluating the CDF of the Gamma RV with parameters $k_{kl}$ and $\theta_{kl}$ at $\gamma$. The corresponding expression is given in (\ref{p_out_kl}).  
\end{appendices}
\color{black}
\newpage
\title{Supplementary material for: Outage Probability Expressions for an IRS-Assisted System with and without Source-Destination Link for the Case of Quantized Phase Shifts in $\kappa - \mu$ Fading}
\maketitle
In this document, we present the simulation results for various special cases of the $\kappa-\mu$ fading. 
\subsection{Results for Rayleigh channel} \label{sec_rayleigh}
First, we consider the case where all the links experience independent Rayleigh fading. Hence we have, $\kappa_{SD}=0$, $\mu_{SD}=1$, $\kappa_{SR}=0$, $\mu_{SR}=1$, $\kappa_{RD}=0$ and $\mu_{RD}=1$. In this case, the uni-variate approximation can be simplified as
\begin{equation}\label{Eq:LIS_DimReducApproxRayleigh}
\begin{aligned}
P_{outage} &\approx 
\frac{4 N d_{sr}^{\beta}d_{rd}^{\beta}}{\gamma_{s}}\int_{0}^{\infty} \left(  1-e^{\frac{-(\sqrt{\gamma}- x - (N-1)\alpha\mu)^{2}d_{sd}^{\beta}}{\gamma_{s}}} \right) \operatorname{U}\left(\sqrt{\gamma}- x - (N-1)\alpha\mu \right)x\textit{K}_0 \left[ \frac{2 x}{\sqrt{\gamma_{s}} d_{sr}^{-\beta/2} d_{rd}^{-\beta/2}}\right] d x
\\
&+N \frac{ 2^{b}}{2\pi}\int_{\frac{-\pi}{2^{b}}}^{\frac{\pi}{2^{b}}} \left(  1-e^{\frac{-\left(\sqrt{\gamma - \left(\alpha\sin{\left( \phi \right)}\mu \right)^{2}}- \alpha\cos{\left( \phi \right)} \mu - (N-1)\alpha\mu\right)^{2}d_{sd}^{\beta}}{\gamma_{s}}} \right) \operatorname{U}\Bigg( \Bigg.\sqrt{\gamma - \left(\alpha\sin{\left( \phi \right)}\mu \right)^{2}}- 
\\& 
\alpha\cos{\left( \phi \right)} \mu - (N-1)\alpha\mu \Bigg.\Bigg) d \phi 
-(2N-1)  \left(  1-e^{\frac{-(\sqrt{\gamma}- N\alpha\mu)^{2}d_{sd}^{\beta}}{\gamma_{s}}} \right) \operatorname{U}\left(\sqrt{\gamma}- N\alpha\mu \right),
\end{aligned}
\end{equation}
where $\mu = \frac{\pi \sqrt{\gamma_{s}}  d_{sr}^{-\beta/2} d_{rd}^{-\beta/2}}{4}$.
\begin{table}[t]
	\begin{tabular}{|l||*{8}{c|}}\hline
		\multirow{5}{*}{}
		$N=5$ &\backslashbox{Method}{Threshold($\gamma$)}
		&\makebox[3em]{-15 dB}&\makebox[3em]{-10 dB}&\makebox[3em]{-5 dB}&\makebox[3em]{-2 dB}
		&\makebox[3em]{0 dB}&\makebox[3em]{2 dB}&\makebox[3em]{5 dB}\\\cline{2-9}
		&Simulated & 0.0937 & 0.2728 &0.6405 & 0.8714 &0.9615 &0.9943 &1.0000 \\\cline{2-9}
		&Uni-variate Approx.  & \textbf{0.0938} & \textbf{0.2732} & \textbf{0.6403} & \textbf{0.8713} & \textbf{0.9615} & \textbf{0.9943} & \textbf{1.0000} \\\cline{2-9}
		&Gamma($k_{mom},\theta_{mom}$) & 0.0947 & 0.2734 & 0.6400 & 0.8712 & 0.9615 & 0.9943 & 1.0000\\\cline{2-9}
		&Gamma($k_{kl},\theta_{kl}$) & 0.0744 & 0.2460     & 0.6328 & 0.8782  & 0.9676 & 0.9961 & 1.0000\\\hline \hline
		$N=50$ &\backslashbox{Method}{Threshold($\gamma$)}
		&\makebox[3em]{-15 dB}&\makebox[3em]{-10 dB}&\makebox[3em]{-5 dB}&\makebox[3em]{-2 dB}
		&\makebox[3em]{0 dB}&\makebox[3em]{2 dB}&\makebox[3em]{5 dB}\\\cline{2-9}
		&Simulated & 0.0543 & 0.2115 & 0.5816 & 0.8402 & 0.9493 & 0.9921 & 1.0000\\\cline{2-9}
		&Uni-variate Approx. & \textbf{0.0545} & \textbf{0.2116} & \textbf{0.5819} & \textbf{0.8403} & \textbf{0.9494} & \textbf{0.9920} & 0.9999\\\cline{2-9}
		&Gamma($k_{mom},\theta_{mom}$) & 0.0635 & 0.2148 & 0.5794 & 0.8394 & 0.9496 & 0.9922 & \textbf{1.0000}\\\cline{2-9}
		&Gamma($k_{kl},\theta_{kl}$) &0.0490 & 0.1909 & 0.5697 & 0.8452 & 0.9559 & 0.9943 & 1.0000 \\\hline 
	\end{tabular}
	\caption{Comparison of OP with SD link for $\mathcal{S}_{1}$ with varying N for $d=30$ and $b=5$}
	\label{diff_N_rayl}
\end{table}

\subsection{Results for Nakagami channel}\label{sec_nakagami}
Next, we consider the case of Nakagami-m fading which is a special case of the $\kappa-\mu$ channel for $\kappa=0$ and $\mu=m$. In this case, the uni-variate approximation can be simplified as

\begin{equation}\label{Eq:LIS_DimReducApproxNakagami}
\begin{aligned}
P_{outage} &\approx  
\frac{4N \left( \sqrt{\frac{a_{SR} a_{RD}}{\gamma_{s}}}\right)^{m_{SR}+m_{RD}}}{\Gamma\left( m_{SR} \right)\Gamma\left( m_{RD}\right)}
\int_{0}^{T_{1}} h_{1}\left(x\right)  
x^{m_{SR}+m_{RD}-1} K_{m_{SR}-m_{RD}}\left(2\sqrt{\frac{a_{SR}a_{RD}}{\gamma_{s}}} x \right) dx
\\
&+N \frac{ 2^{b}}{2\pi}\int_{\frac{-\pi}{2^{b}}}^{\frac{\pi}{2^{b}}}
\begin{array}{l}
P\left(m_{SD}, \frac{m_{SD}\left(\sqrt{\gamma - \left(\alpha\sin{\left( \phi_{i} \right)}\mu \right)^{2}}- 
	\alpha\cos{\left(\phi_{i}\right)} \mu - (N-1)\alpha\mu\right)^{2}}{\gamma_{s}\hat{t}^{2}_{SD}} \right)   \\
\times\operatorname{U}\Bigg( \Bigg.\sqrt{\gamma - \left(\alpha\sin{\left( \phi \right)}\mu \right)^{2}}-
\alpha\cos{\left( \phi \right)} \mu - (N-1)\alpha\mu \Bigg.\Bigg)  d \phi    
\end{array}
\\
&-(2N-1) P\left(m_{SD},\frac{m_{SD}\left(\sqrt{\gamma} - N\alpha\mu\right)^{2}}{\gamma_{s}\hat{t}^{2}_{SD}}\right) \operatorname{U}\left(\sqrt{\gamma}- N\alpha\mu \right),
\end{aligned}
\end{equation}
where
\begin{equation}
\begin{aligned}
a_{SR} &= \frac{m_{SR}}{\hat{t}^{2}_{SR}}, a_{RD} = \frac{m_{RD}}{\hat{t}^{2}_{RD}}, \quad 
T_{1} = \sqrt{\gamma} - \left(N-1\right)\alpha\mu, \\
\mu &= \frac{\sqrt{\gamma_{s}}\left( m_{SR} \right)_{\frac{1}{2}} \left( m_{RD} \right)_{\frac{1}{2}}}{\sqrt{a_{SR} a_{RD}}}, \quad        h_{1}\left(x\right) = P\left(m_{SD},\frac{m_{SD}\left(\sqrt{\gamma} - x - (N-1)\alpha\mu\right)^{2}}{\gamma_{s}\hat{t}^{2}_{SD}}\right).
\end{aligned}
\end{equation}
and $\textit{K}_{\nu}$ is modified Bessel function of the second kind of order $\nu$ \cite{BesselKdef}, $P\left(a,z\right)$ is the Regularized Gamma function \cite{ReglowGammadef} and $\operatorname{U}(\cdot)$ is the unit step function. \par Here, the results in Table \ref{diff_N_Nakagami} are generated using the following parameters: $\kappa_{SD}=0$, $\mu_{SD}=1$, $\kappa_{SR}=0$, $\mu_{SR}=2$, $\kappa_{RD}=0$ and $\mu_{RD}=1.2$. Note that this simulation setting is similar to the scenario considered in Table \ref{diff_N_rayl} except that the values of $\mu_{SR}$ and $\mu_{RD}$ are larger. We can observe that the OP decreases with the increase in the number of multipath clusters in the $SR$ and $RD$ link. This decrease in OP is more significant for large values of $N$. 
\begin{table}[t]
	\begin{tabular}{|l||*{8}{c|}}\hline
		\multirow{5}{*}{}
		$N=5$ &\backslashbox{Method}{Threshold($\gamma$)}
		&\makebox[3em]{-15 dB}&\makebox[3em]{-10 dB}&\makebox[3em]{-5 dB}&\makebox[3em]{-2 dB}
		&\makebox[3em]{0 dB}&\makebox[3em]{2 dB}&\makebox[3em]{5 dB}\\\cline{2-9}
		&Simulated & 0.0936 & 0.2721 & 0.6401 & 0.8711 &0.9614 & 0.9944 & 1 \\\cline{2-9}
		&Uni-variate Approx.  & \textbf{0.0934} & \textbf{0.2726} & \textbf{0.871} & \textbf{0.9244} & \textbf{0.9614} & \textbf{0.9943} & \textbf{1} \\\cline{2-9}
		&Gamma($k_{mom},\theta_{mom}$) &0.0944  & 0.2728 & 0.6395 &0.8709  & 0.9614 & 0.9943 & 1 \\\cline{2-9}
		&Gamma($k_{kl},\theta_{kl}$) & 0.0741  &0.2455   & 0.6323 & 0.8779  &0.9675  & 0.9961 & 1 \\\hline \hline
		$N=50$ &\backslashbox{Method}{Threshold($\gamma$)}
		&\makebox[3em]{-15 dB}&\makebox[3em]{-10 dB}&\makebox[3em]{-5 dB}&\makebox[3em]{-2 dB}
		&\makebox[3em]{0 dB}&\makebox[3em]{2 dB}&\makebox[3em]{5 dB}\\\cline{2-9}
		&Simulated &0.0511 &0.2064  &0.5755  & 0.8376 & 0.9483 & 0.9919 & 1\\\cline{2-9}
		&Uni-variate Approx. & \textbf{0.0514} & \textbf{0.2063} & \textbf{0.5765} & \textbf{0.8373}  & \textbf{0.9482} & \textbf{0.9917} & \textbf{0.9999}\\\cline{2-9}
		&Gamma($k_{mom},\theta_{mom}$)  & 0.0611 & 0.2098 & 0.5737 & 0.8363 & 0.9483 & 0.9919 & 1\\\cline{2-9}
		&Gamma($k_{kl},\theta_{kl}$)  & 0.047 & 0.1862 & 0.5639 & 0.8419 & 0.9547 & 0.9941 &1  \\\hline 
	\end{tabular}
	\caption{Comparison of OP with SD link for $\mathcal{S}_{1}$ with varying N for $d=30$ and $b=5$}
	\label{diff_N_Nakagami}
\end{table}

\subsection{Results for Rician channel}\label{sec_rician}
Here, we consider the case of Rician fading with parameter $K$, which is a special case of the $\kappa-\mu$ channel for $\kappa=K$ and $\mu=1$. The results in Table \ref{diff_N_Rician} are generated using the following parameters: we have $\kappa_{SD}= 0$, $\mu_{SD}=1$,$\kappa_{SR}=2$, $\mu_{SR}=1$, $\kappa_{RD}=2.5$ and $\mu_{RD}=1$. This setting is similar to the scenario considered in Table \ref{diff_N_rayl} except that the values of $\kappa_{SR}$ and $\kappa_{RD}$ are larger. Note that the parameter $\kappa$ represents the ratio  between  the  total  power  of the  dominant components and  the  total  power  of the  scattered  waves. Thus, larger values of $\kappa_{SR}$ and $\kappa_{RD}$ means that the dominant components of the $SR$ and $RD$ links are stronger and hence the OP decreases. This can be observed from the results in Table \ref{diff_N_Rician}. 
\begin{table}[t]
	\begin{tabular}{|l||*{11}{c|}}\hline
		\multirow{3}{*}{}
		$N$ &\backslashbox{Method}{Threshold($\gamma$)}
		&\makebox[3em]{-15 dB}&\makebox[3em]{-10 dB}&\makebox[3em]{-5 dB}&\makebox[3em]{-2 dB}
		&\makebox[3em]{0 dB}&\makebox[3em]{2 dB}&\makebox[3em]{5 dB}\\\hline \hline
		$N=5$ &Simulated  & $0.0932$ & $0.2720$ & $0.6386$ & $0.8707$ &  $0.9610$ & $0.9942$ & $1.0000$ \\\cline{2-9}
		&Uni-variate Approx. & $\mathbf{0.0933}$ & $0.2724$ & $0.6396$ & $0.8709$ &  $0.9613$ & $0.9943$ & $1.0000$ \\\cline{2-9}
		&Gamma($k_{mom},\theta_{mom}$)  & $0.0943$ & $0.2726 $ & $ 0.6393$ &  $0.8709$ & $0.9614$  & $0.9943$  & $1.0000$ \\\cline{2-9}
		&Gamma($k_{kl},\theta_{kl}$) & $0.0740$ &  $0.2453$& $0.6321$  & $0.8779$ &  $0.9675$ & $0.9961$  & $1.0000$ \\\hline \hline
		$N=50$ &Simulated  & $0.0500$ & $0.2044$  & $0.5757$  & $0.8368$  & $0.9479$  & $0.9917$  & $1.0000$ \\\cline{2-9}
		&Uni-variate Approx.  & $0.0505$ & $0.2047$  & $0.5748$  & $0.8364$  & $0.9478$  & $0.9916$  & $0.9999$ \\\cline{2-9}
		&Gamma($k_{mom},\theta_{mom}$)  & $0.0604$ & $0.2084$ & $0.5720$  & $0.8353$ & $0.9480$  & $0.9919$ & $0.9999$ \\\cline{2-9}
		&Gamma($k_{kl},\theta_{kl}$) & $0.0465$ & $0.1849$ & $0.5622$  & $0.8410$ & $0.9543$  & $0.9940$ & $1.0000$ \\\hline \hline
	\end{tabular}
	\caption{Comparison of OP with SD link for $\mathcal{S}_{1}$ with varying N for $d=30$ and $b=5$}
	\label{diff_N_Rician}
\end{table}
\subsection{Results for $\kappa-\mu$ channel}\label{sec_kappamu}
Here, we present the results for the most general case of $\kappa-\mu$ fading scenario with the following parameters: $\kappa_{SD}=0.7$, $\mu_{SD}=1$, $\kappa_{SR}=2.4$, $\mu_{SR}=5$, $\kappa_{RD}=1.8$ and $\mu_{RD}=3$. From the results in Section\ref{sec_rayleigh} to Section\ref{sec_rician}, we know that the increase in $\kappa_{SR},\kappa_{RD}$, $\mu_{SR}$ and $\mu_{RD}$ decreases OP. Here, the value of $\kappa_{SD}$ is also larger than the value of $\kappa_{SD}$ used in Section \ref{sec_rayleigh}. From the results in Table \ref{diff_N_k_mu}, we can see that the OP is lesser than the values of OP observed in Table \ref{diff_N_rayl}.  
% \begin{table}[t]
% \begin{tabular}{|l||*{11}{c|}}\hline
% \multirow{3}{*}{}
% $N$ &\backslashbox{Method}{Threshold($\gamma$)}
% &\makebox[3em]{-15 dB}&\makebox[3em]{-10 dB}&\makebox[3em]{-5 dB}&\makebox[3em]{-2 dB}
% &\makebox[3em]{0 dB}&\makebox[3em]{2 dB}&\makebox[3em]{5 dB}\\\hline \hline
% $N=5$ &Simulated &0.1471   & 0.3247     &0.6439 &   0.8569    & 0.9523    &0.9922 &1.0000 \\\cline{2-9}
% &Uni-variate Approx.  &0.1464& 0.3233 &0.6432 &0.8565    &0.9516    &0.9914 &0.9993 \\\cline{2-9}
% &Gamma($k_{mom},\theta_{mom}$) &0.1255    &0.3097   &0.6493 &0.8624    &0.9528    &0.9912 &0.9999\\\cline{2-9}
% &Gamma($k_{kl},\theta_{kl}$) &0.0986   & 0.2779   & 0.6404   & 0.8692   & 0.9600   & 0.9939   & 1.0000\\\hline \hline
% $N=50$ &Simulated &0.0892   & 0.2546   &0.5811   &0.8190  &0.9352   &0.9882  &0.9999 \\\cline{2-9}
% &Uni-variate Approx. &0.0890& 0.2536 &0.5802  &0.8190   &0.9347 & 0.9876 &0.9995 \\\cline{2-9}
% &Gamma($k_{mom},\theta_{mom}$) &0.0834    &0.2409   &0.5819   &0.8246  &0.9366   &0.9874  &0.9998 \\\cline{2-9}
% &Gamma($k_{kl},\theta_{kl}$) &0.0642  & 0.2130  &0.5701   &0.8297  &0.9439   &0.9906  &0.9999 \\\hline \hline
% \end{tabular}
%  \caption{Comparison of OP with SD link for $\mathcal{S}_{1}$ with varying N for $d=30$ and $b=5$}
% \label{diff_N_k_mu}
% \end{table}

\begin{table}[t]
	\begin{tabular}{|l||*{11}{c|}}\hline
		\multirow{3}{*}{}
		$N$ &\backslashbox{Method}{Threshold($\gamma$)}
		&\makebox[3em]{-15 dB}&\makebox[3em]{-10 dB}&\makebox[3em]{-5 dB}&\makebox[3em]{-2 dB}
		&\makebox[3em]{0 dB}&\makebox[3em]{2 dB}&\makebox[3em]{5 dB}\\\hline \hline
		$N=5$ &Simulated &0.0797   & 0.2440  &  0.6222  &  0.8772  &  0.9705  &  0.9973   & 1.0000 \\\cline{2-9}
		&Uni-variate Approx.  & 0.0795 & 0.2437 & 0.6213 & 0.8766 & 0.9699  & 0.9967 & 0.9993 \\\cline{2-9}
		&Gamma($k_{mom},\theta_{mom}$) &0.0653   & 0.2322  &  0.6280   & 0.8808  &  0.9700   & 0.9968 &   1.0000\\\cline{2-9}
		&Gamma($k_{kl},\theta_{kl}$) &0.0512  &  0.2094   & 0.6222  &  0.8876   & 0.9748   & 0.9978   & 1.0000\\\hline \hline
		$N=50$ &Simulated & 0.0389  &  0.1736 &   0.5441  &  0.8350  &  0.9560   & 0.9954  &  1.0000 \\\cline{2-9}
		&Uni-variate Approx. & 0.0388 & 0.1737 & 0.5432 & 0.8344  & 0.9556 & 0.9950 &0.9996 \\\cline{2-9}
		&Gamma($k_{mom},\theta_{mom}$) &0.0364   & 0.1637  &  0.5460  &  0.8390  &  0.9561 &   0.9948  &  1.0000 \\\cline{2-9}
		&Gamma($k_{kl},\theta_{kl}$) &0.0279  &  0.1453   & 0.5374  &  0.8446  &  0.9613  &  0.9962  &  1.0000 \\\hline \hline
	\end{tabular}
	\caption{Comparison of OP with SD link for $\mathcal{S}_{1}$ with varying N for $d=30$ and $b=5$}
	\label{diff_N_k_mu}
\end{table}
\bibliographystyle{IEEEtran}
\bibliography{reference.bib}
\end{document}